\documentclass[11pt]{article}
\usepackage{geometry}
\usepackage{amsmath,mathtools,amssymb}
\usepackage{amsfonts}
\usepackage{bm}
\usepackage{amsthm}
\usepackage{enumerate}
\usepackage{esint}
\usepackage{color}
\usepackage [english]{babel}
\usepackage [autostyle, english = american]{csquotes}
\usepackage[T1]{fontenc}
\usepackage[utf8]{inputenc}
\usepackage{authblk}
\usepackage{graphicx}
\usepackage{subcaption}
\numberwithin{equation}{section}
\usepackage{titlesec}
\usepackage{tocloft}

\titleformat*{\section}{\centering \large\bfseries}
\titleformat*{\subsection}{\centering \large\itshape}

\MakeOuterQuote{"}
\DeclareMathOperator{\Tr}{Tr}

\DeclareMathOperator{\rank}{rank}

\DeclareMathOperator{\cof}{cof}

\DeclareMathOperator{\II}{II}

\allowdisplaybreaks

\newtheorem{theorem}{Theorem}[section]

\newtheorem{proposition}[theorem]{Proposition}

\title{
Microstructure-enabled control of wrinkling in nematic elastomer sheets
}
\author{Paul Plucinsky and Kaushik Bhattacharya} 
\affil{California Institute of Technology, Pasadena, California 91125, USA}

\date{\today}
\begin{document}
\maketitle

\begin{abstract} 
Nematic elastomers are rubbery solids which have liquid crystals incorporated into their polymer chains.  These materials display many unusual mechanical properties, one such being the ability to form fine-scale microstructure.  In this work, we explore the response of taut and appreciably stressed sheets made of nematic elastomer.  Such sheets feature two potential instabilities -- the formation of fine-scale material microstructure and the formation of fine-scale wrinkles.  We develop a theoretical framework to study these sheets that accounts for both instabilities, and we implement this framework numerically.  Specifically, we show that these instabilities occur for distinct mesoscale stretches, and observe that microstructure is finer than wrinkles for physically relevant parameters.  Therefore, we relax (i.e., implicitly but rigorously account for) the microstructure while we regularize (i.e., compute the details explicitly) the wrinkles.  Using both analytical and numerical studies, we show that nematic elastomer sheets can suppress wrinkling by modifying the expected state of stress through the formation of microstructure.
\end{abstract}

%%%%%%%%%%%%%%%%%%%%%%%%%%%%%%%%%%%%%%%%%
%%%%%%%%%%%%%%%%%%%%%%%%%%%%%%%%%%%%%%%%%
\section{Introduction}

Nematic elastomers are rubbery solids made of cross-linked polymer chains that have nematic mesogens (rod-like molecules) either incorporated into the main chain or pendent from them.  Their structure enables a coupling between the entropic elasticity of the polymer network and the ordering of the liquid crystals, and this in turn results in fairly complex mechanical properties (see Warner and Terentjev \cite{wt_lceboox_03} for a comprehensive introduction and review).  At high temperatures, the mesogens are randomly oriented and the material is isotropic.  However, on cooling below a critical temperature, the mesogens undergo an isotropic to nematic phase transformation and develop a local orientational order described by a director.  This is accompanied by a spontaneous elongation along the director and contraction transverse to it.  However, since the material is isotropic in the high temperature state, there is no preferential orientation for the director. Thus, the director may rotate freely with respect to the material frame and form domains where the director varies spatially.  This manifests itself in a very rich range of phenomena.

\begin{figure}
\centering
\includegraphics[height = 3.0 in]{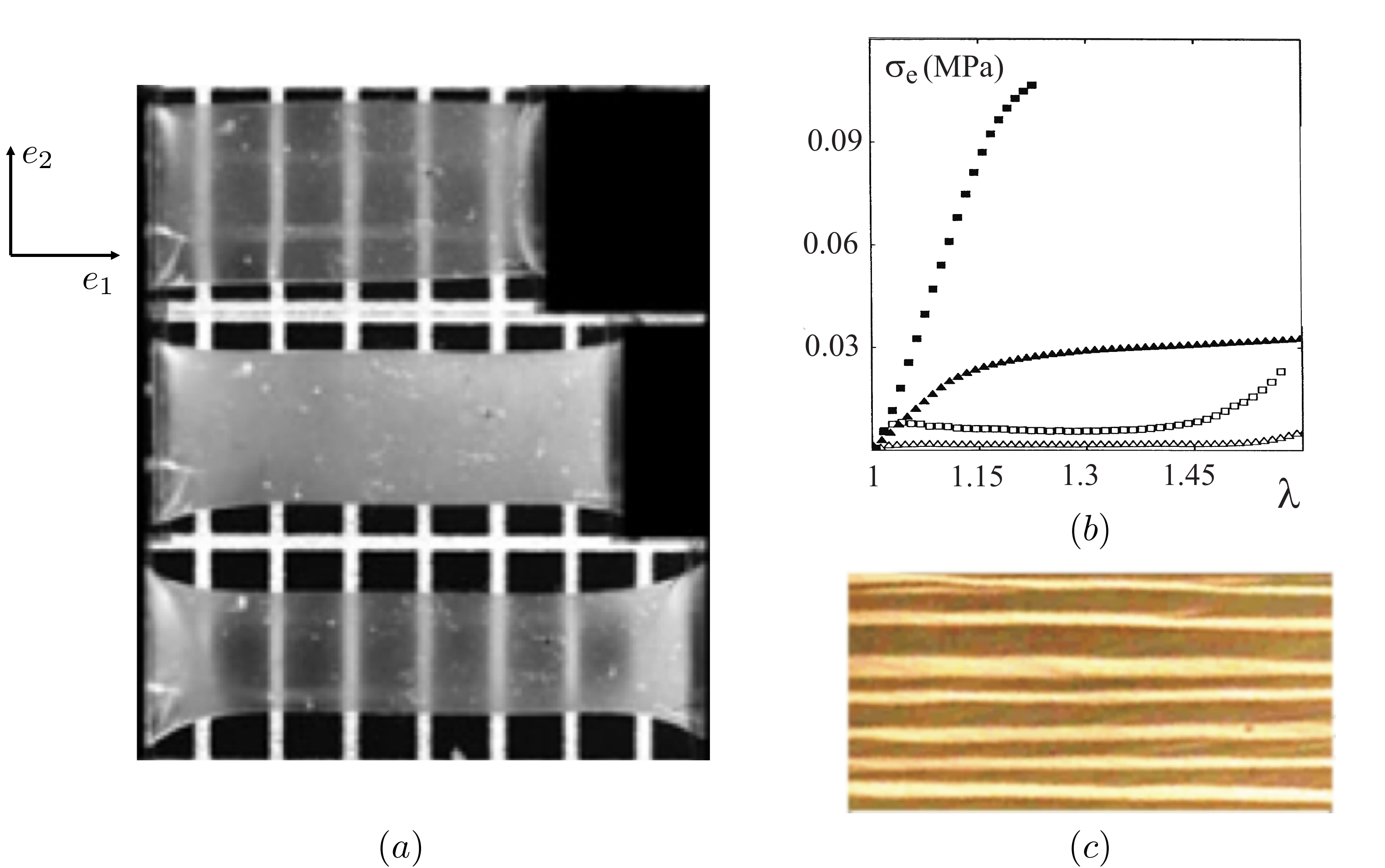}
\caption{The clamped-stretch experiments of Kundler and Finkelmann \cite{kf_mrc_95} showing soft elasticity and fine scale microstructure, but no wrinkling. (a) Snapshots of the sheet.  The nematic director is oriented vertically when undeformed (top), develops stripe domains of alternating rotated directors depicted in (c) for moderate deformation (middle) and eventually is uniform and oriented horizontally (bottom). (b)  Stress-strain curves  of K\"{u}pfer and Finkelmann \cite{kupf_macro_94} for elastomers of different preparation histories -- the lowest curve is akin to the clamped-stretch sheet in part (a) and describes a sheet with significant soft elasticity. (c) Fine-scale strip domain microstructure. }
\label{fig:CSKF}
\end{figure}
Of particular interest is the soft elasticity and fine scale microstructure (textured deformation or striped domains) observed in the clamped-stretch experiments Kundler and Finkelmann \cite{kf_mrc_95}.  
Some of their key observations are reproduced in Figure \ref{fig:CSKF}.  
They begin with a thin rectangular sheet where the director is uniformly oriented tangential to the sheet but in the short direction (top of Figure \ref{fig:CSKF}(a)).  The fact that the director is uniform is evident from the fact that the sheet is transparent.  They clamp the short edges and pull it along the long edge.  The nominal stretch vs. nominal stress behavior for this sheet is akin the bottom-most curve in Figure \ref{fig:CSKF}(b).  Notice that nematic sheets display essentially zero stress for very significant values of stretch: this is known as the {\it soft elasticity}.  The center and bottom images of Figure \ref{fig:CSKF}(a) capture stretch midway through the soft behavior and at the end of the soft behavior respectively.  During the soft stage of stretch, the entire sheet is strongly scattering light and so visibly cloudy, an optical indication that the director is no longer uniform.  Beyond the soft plateau (i.e., the bottom image), the sheet becomes transparent again in its central region, indicating uniform director arrangement in this region.  However, it remains cloudy near the clamped edges.

The cloudy regions indicate material microstructure in the form of strip domains of oscillating director as shown in Figure \ref{fig:CSKF}(c).  The heuristic is as follows: A nematic elastomer features a director that can rotate through the material, and this rotation is accommodated by spontaneous elongation along the director.  Thus, the sheet can elongate along the direction of stretch with little stress by having its director rotate from vertical to horizontal.  Doing so uniformly, however, results in a shear at intermediate orientations, but this shear can be avoided on average by breaking up the cross-section into domains where one half the directors rotate one way (say through $\theta$) while the other half rotates the other way (through $- \theta$).  This is exactly what happens in the clamped stretched sheet at a very fine scale (microns), manifesting in stripe domains.  Finally, when the director has fully rotated to the horizontal, it becomes uniform again (since the material is invariant under the change of sign of the director).

Bladon {\it {\it et al.}\ }\cite{btw_pre_93} proposed a free energy based on entropic elasticity of the chains in the presence of nematic order to describe the elasticity of nematic elastomers, and Verwey {\it {\it et al.}\ }\cite{vwt_jphys_96} explained how stripe domains can arise as a means of minimizing this free energy.  
DeSimone and Dolzmann \cite{dd_arma_02} noted that the free energy density proposed by Bladon {\it {\it et al.}\ }\cite{btw_pre_93} is not quasiconvex, and thus fine-scale microstructure can arise naturally in these materials.  These include stripe domains, but also more complex microstructure.  They also computed explicitly the relaxation of the Bladon {\it {\it et al.}\ }free energy which implicitly but rigorously accounts for the microstructure.  Conti {\it {\it et al.}\ }\cite{cdd_02_pre, cdd_02_jmps} used the planar version of free energy to study the stretching of sheets and were able to explain various details of the experiments described above including the soft elasticity, formation and disappearance of stripe domains, and the persistence of domains near the grips even at high stretches.

In this paper, we study another surprising observation inherent in these experiments, one that that has thus far escaped notice and exploration.  Even though these thin sheets have been stretched significantly with clamped grips, they remain planar and do not wrinkle.  In fact, similar experiments have been conducted by a number of researchers, and none of them have reported any wrinkling instability.  This is surprising because thin sheets of purely elastic materials wrinkle readily when subjected to either shear \cite{wp_06_jmms} or stretching with clamped grips \cite{nrh_11_ijss,tbs_14_jmps,z_08_wrinkle}. 

The wrinkling of thin elastic membranes has been widely studied, motivated by various application.  Early research was motivated by the use of membranes for aircraft skins where wrinkling altered their aerodynamic performance.  More recent interest stems from the use of membranes in light-weight deployable space structures including solar sails, telescopes and antennas \cite{j_paa_01,m_spring_13}, and renewed interest in fabric roofs of complex shape \cite{bgb_se_04} (see also \cite{z_08_wrinkle} and references therein).  The underlying mechanism is relatively simple: thin elastic membranes are unable to sustain any compression; instead they accommodate imposed compressive strains by buckling out of plane.  When a sheet is pulled on clamped edges, the clamps inhibit the natural lateral contraction, leading to compressive stresses in the lateral direction, which in turn leads to wrinkles or undulations elongated along the direction of stretch.  The wavelength of the undulations are large compared to thickness, but small compared to the overall dimensions of the sheet.

Mathematically, any finite deformation theory of membranes is not quasiconvex, and thus suffers from instabilities which can be interpreted as wrinkles (see for example  \cite{p_86_jam,s_90_prsa,sp_89_qjmam}).
Further, the relaxation of such theories gives rise to tension-field theories like those of Mansfield  \cite{m_70_prsa} where membranes can resist tension but not compression.  Such theories are zero thickness idealizations of the membrane which account for the consequences of wrinkles at a scale large compared to the wrinkles but do not describe them explicitly.  Alternatively, in recognizing that wrinkles cause bending due to the non-zero thickness of the membrane, Koiter-type theories, which capture a sum  of bending and membranes energies, lead to an explicit description of wrinkles.  Such theories form the basis of the analysis of wrinkling described above \cite{nrh_11_ijss,tbs_14_jmps, wp_06_jmms,z_08_wrinkle}.

In general, there are two approaches in dealing with instabilities resulting from the failure of (an appropriate notion of) convexity that results in features at a fine scale.  The first is relaxation, where one derives an {\it effective} or {\it relaxed} theory that describes the overall behavior after accounting for the formation of fine-scale features.  The relaxed theory of DeSimone and Dolzmann in the context of liquid crystal elastomers and the tension-field theories for thin membranes are examples of such relaxation.   
While these theories are extremely useful in describing overall behavior, they are often difficult to compute explicitly and they do not resolve all fine scale details though it is at times possible to {\it a posteriori} reconstruct them.  Further, they are often degenerately convex and therefore lead to extremely stiff numerical problems.  The second approach is {\it regularization} where one recognizes that the source of the nonconvexity is the neglect of some smaller order physics, and adds some higher order term to the energy.  The second gradient theories of plasticity, the phase field theories of phase transformations and the theories of Verwey{\it {\it et al.}\ } \cite{vwt_jphys_96} (where the Frank elasticity regularizes the entropic elasticity) are examples.  These resolve the fine-scale details, but are computationally extremely expensive as they require a very fine resolution.

In this work, we are interested in the potential wrinkling behavior of stretched nematic elastomer sheets.  Therefore, we have to account for two sources of instability -- a material instability that results in the formation of fine-scale microstructure and a structural instability that results in fine-scale wrinkles.  We note that the scale of the microstructure (microns) is small compared to the scale of wrinkles (millimeters).   Therefore, we take a multiscale view and systematically develop a theory that is a relaxation for microstructure but a regularization for wrinkles.   The resulting theory is a Koiter-type theory (\ref{eq:Koiter}) with two terms; the first is the two-dimensional or plane stress reduction of the relaxed energy of DeSimone and Dolzmann  \cite{dd_arma_02} and the second is bending.

To develop this theory, we build on the work of Cesana {\it {\it et al.}\ }\cite{cpk_arma_15} who started from an appropriate three dimensional formulation and derived from it the relaxed membrane energy.  They also provided the explicit characterization of the instabilities or oscillations (Young measures) that underly the relaxation.  This had two sources -- microstructure or stripe domains and wrinkles.  Remarkably, for the taut sheets of current interest, they found that the overall deformation gradients for which microstructure occurs are distinct from the overall deformation gradients for which wrinkling occurs.  Here, we show, in addition, that the plane stress reduction of the relaxed entropic elastic energy coincides with the relaxed membrane energy in all regions of interest for taut membranes except the one involving tension wrinkles, where it is the plane stress reduction of the original entropic energy.   Consequently, this reduction accurately describes the role of microstructure in the in-plane deformation of nematic elastomer sheets.  It does not, however, accurately describe tension wrinkling in these sheets; rather, regularization or relaxation is needed.  Taking the regularization approach, we use the Young measure characterization of tension wrinkling oscillations to compute the bending energy for these oscillations systematically from the relaxed entropic elastic energy.  This bending together with the plane stress reduction of the relaxed entropic elastic energy gives the appropriate Koiter-type theory for taut sheets of nematic elastomer.

We use this theory in numerical studies to demonstrate that the ability of the material to form microstructure does indeed suppress wrinkling.  Specifically, we study the clamped-stretch experiments of Kundler and Finkelmann \cite{kf_mrc_95}  and focus on sheets with lateral dimensions for which purely elastic materials readily wrinkle under this stretch.  We show that as a parameter that describes the strength of the nematic order increases, the onset of wrinkling is delayed and the amplitude is decreased, until it is completely suppressed for large enough nematic order.  We further show that the reason for this is that the ability to form microstructure alters the stress distribution close to the clamps.  These results open up the possibility of exploiting these materials in applications where one seeks membranes that do not wrinkle.

This paper is organized as follows:  In section \ref{sec:Hierarchy}, we comment on the notation and present for clarity a visual summary capturing  the hierarchy of theories for nematic elastomers described and developed in sequel. We provide the background in Section \ref{sec:TheoryBack}, recalling the theory of Bladon {\it {\it et al.}\ }\cite{btw_pre_93} for the entropic elasticity of nematic elastomers (Section \ref{ssec:FreeEnergy}), the relaxation of this energy by DeSimone and Dolzmann  \cite{dd_arma_02} (Section \ref{ssec:3Dsec}) and the membrane theory of Cesana {\it et al.} \cite{cpk_arma_15}  (Section \ref{ssec:EffectiveTheory}).  We develop the Koiter theory in Sections \ref{sec:k}.  We describe the numerical method in Section \ref{sec:Numerical}, and apply it to study clamped extension of sheets in Section \ref{sec:results}.  We conclude in Section \ref{sec:conc}.

\section{Notation and Overview}\label{sec:Hierarchy}

\begin{figure*}
\centering
\begin{subfigure}{\textwidth}
\centering
\includegraphics[width = 5.0 in]{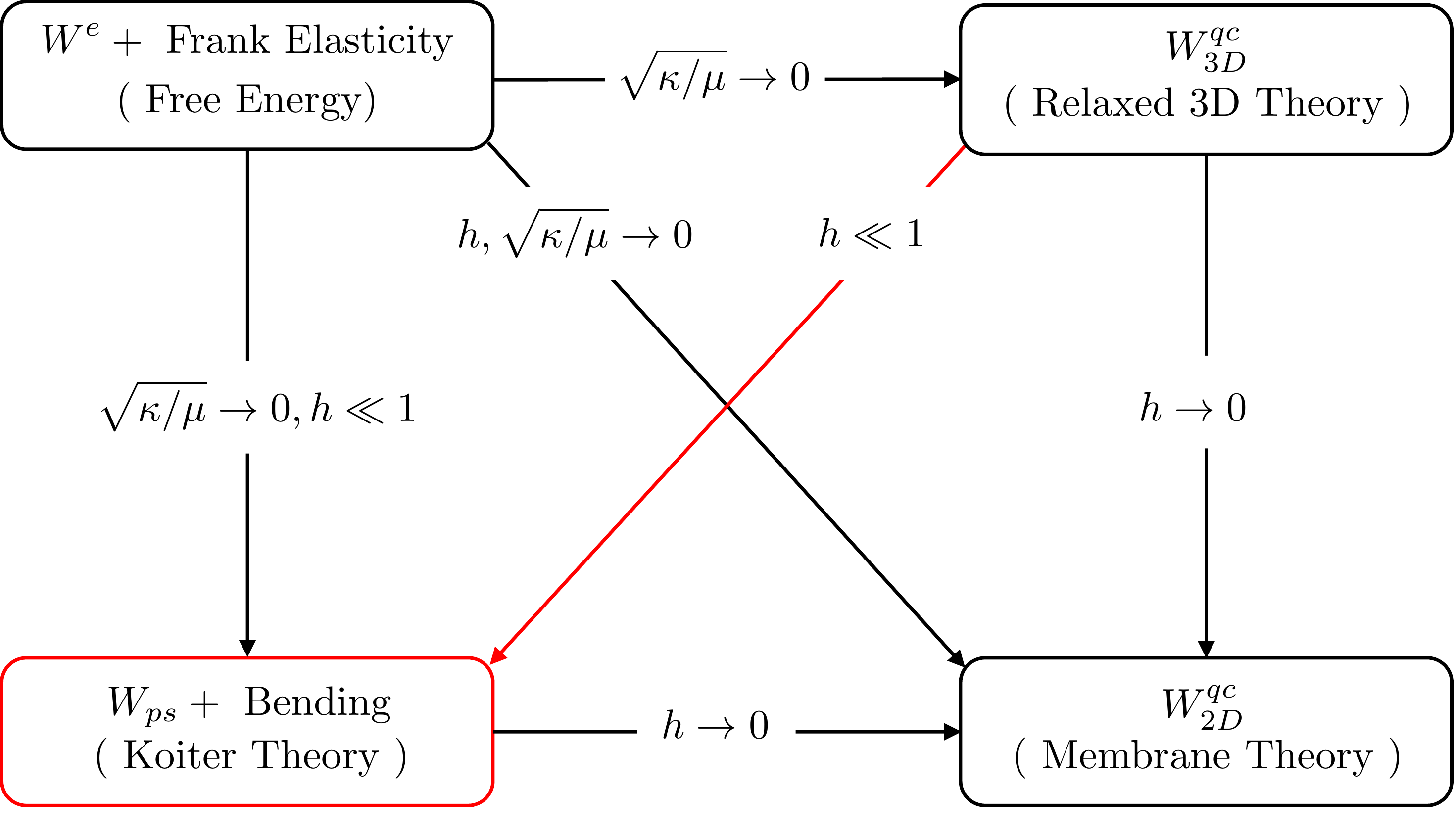}
\caption{Hierarchy of theories for nematic elastomers.}
\label{fig:Schematic1}
\end{subfigure}
\begin{subfigure}{\textwidth}
\vspace*{.2cm}
\centering
\includegraphics[width = 5.0 in]{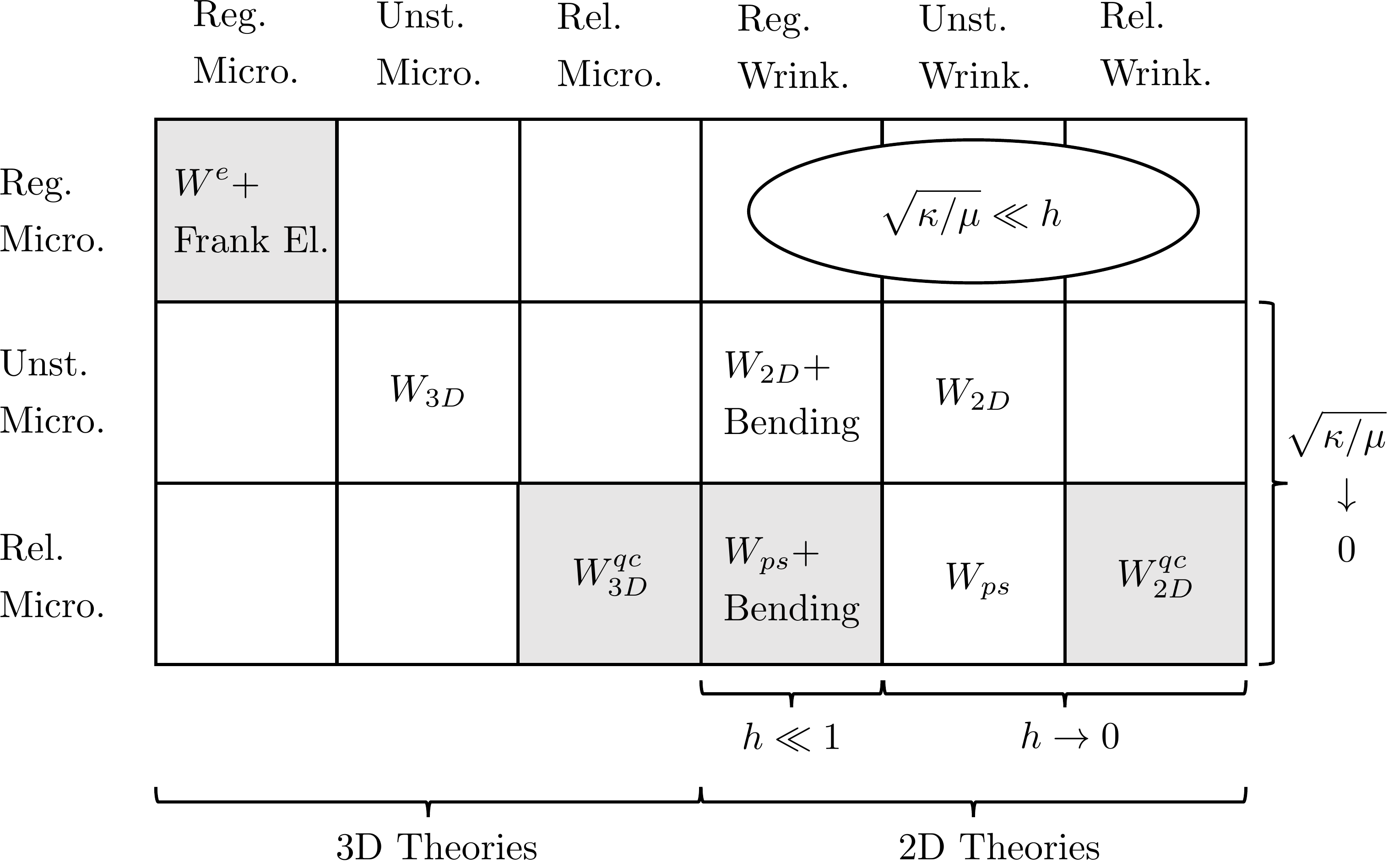}
\caption{Strain energy densities and their relation to instabilities.}
\label{fig:Schematic2}
\end{subfigure}
\caption{(a) Theories which properly account for the formation of microstructure in 3D and both microstructure and wrinkling in 2D.  Each theory can be derived in the treatment of the lengthscale as depicted. We focus here on deriving the Koiter theory (highlighted in red) starting from the relaxed 3D theory. Elsewhere \cite{p_thesis}, we derive it starting form the free energy, where there are simply more details to track. (b) The various strain energy densities of nematic elastomers and whether they are stable $-$ regularized (reg.)  or relaxed (rel.) $-$ or unstable (unst.) to instabilities: microstructure (micro.) in the case of 3D densities, and both microstructure and wrinkling (wrink.) in the case of 2D densities.  The stable theories are highlighted.  Starting from any strain energy density on this chart, the densities to the right, lower diagonal, and directly below can be derived under the treatment of the lengthscales as depicted.  In principal, the top right of this chart can be populated, though such theories are less physically relevant given the disparity in lengthscales.}
\label{fig:EnergiesCartoon}
\end{figure*}

\subsection{Some of the notation}
We denote with $\mathbb{R}^n$ the $n$ dimensional Euclidian space endowed with the usual scalar product $u \cdot v := u^T v$ and norm $|u|:= \sqrt{u \cdot u}$. We denote the unit sphere in $\mathbb{R}^n$ by $\mathbb{S}^{n-1}$ and it is defined as the set of all vectors $u \in \mathbb{R}^n$ with $|u| = 1$. We label with $\mathbb{R}^{m\times n}$ the space of $m \times n$ matrices with real entries. For $n > 1$, we denote with $SO(n)$  the space of rotation matrices (i.e., each $F \in \mathbb{R}^{n\times n}$ such that $F^TF = I$ and $\det F = 1$).  We take $\mathbb{R}^+$ to be the set of non-negative real numbers.  

We often describe the material points of a three dimensional solid with $x := x_1e_1 + x_2 e_2 + x_3 e_3$ for the fixed right-handed orthonormal basis $\{e_1, e_2, e_3\} \subset \mathbb{R}^3$ depicted in Figure \ref{fig:CSKF}.  Similarly, we often describe the material points of a two dimensional sheet with $\tilde{x} := x_1 \tilde{e}_1 + x_2 \tilde{e}_2$ for the analogous two dimensional orthonormal basis $\{\tilde{e}_1, \tilde{e}_2\} \subset \mathbb{R}^2$.   As with the basis vectors and material points, we use tilde as a mean to distinguish between two dimensional and three dimensional quantities (if there is no conflict with previous notation). So we use $F \in \mathbb{R}^{3\times3}$ to describe the deformation gradient of a solid and $\tilde{F} \in \mathbb{R}^{3\times2}$ to describe the planar deformation gradient of a sheet; we denote with $\nabla$  the three dimensional gradient (with respect to $x$) and $\tilde{\nabla}$ the planar gradient (with respect to $\tilde{x}$); $\ldots$; etc.  

Lastly, we find it natural at points to introduce certain mathematical concepts: $W^{1,p}$ Sobolev Spaces and weak convergence (i.e., $\rightharpoonup$) in these spaces, quasiconvexification as a means of relaxation, $\Gamma$-convergence as a means of dimension reduction and the theory of gradient Young measures for characterizing instabilities.  We refer to Evans \cite{e_pde_98}, Dacorogna \cite{dac_book_07}, Braides \cite{b_02_gcbook} and M\"{u}ller \cite{m_99_cov} respectively for introductions into these concepts.

\subsection{Overview on the hierarchy of theories for nematic elastomers}
In this work, we systematically develop a two dimensional Koiter theory for nematic elastomer sheets by both (i) starting from an appropriate three dimensional description of these elastomers \cite{btw_pre_93,dd_arma_02,vwt_jphys_96,wt_lceboox_03}, and (ii) building off of the development of the effective or relaxed membrane theory \cite{cpk_arma_15}.  Thus, in the course of this development, we find it natural to introduce several variants of strain energy densities modeling nematic elastomers, for which their is a well-characterized hierarchy.  The hierarchy is related to the mathematical treatment of small-length scales (i.e., $\sqrt{\kappa/\mu}$ and $h$) inherent to nematic elastomer sheets, and how instabilities (both wrinkling and microstructure) are accounted for in this treatment.  

Briefly (all of this is expanded upon in Sections \ref{sec:TheoryBack} and \ref{sec:k}), a three dimensional nematic elastomer can form fine-scale microstructure on a lengthscale $\sqrt{\kappa/\mu}$ related to the competition between entropic and Frank elasticity in these solids.  Here, $\mu$ is the shear modulus of the polymer network, and $\kappa$ is the characteristic modulus penalizing deviations in the average ordering of liquid crystals away from the preferred spatially uniform alignment.  In addition to this microstructure, a sheet of nematic elastomer may wrinkle since the thickness $h$ is small compared to the lateral extent of the sheet.  In typical sheets, $\sqrt{\kappa/\mu} \ll h$.  To guide the reading of Sections \ref{sec:TheoryBack} and \ref{sec:k}, we provide a visual summary of the theories which emerge in the competition of these two lengthscales and their hierarchy (Figure \ref{fig:EnergiesCartoon}).

%%%%%%%%%%%%%%%%%%%%%%%%%%%%%%%%%%%%%%%%%
%%%%%%%%%%%%%%%%%%%%%%%%%%%%%%%%%%%%%%%%%
\section{Background}\label{sec:TheoryBack}

%%%%%%%%%%%%%%%%%%%%%%%%%%%%%%%%%%%%%%%%%
%%%%%%%%%%%%%%%%%%%%%%%%%%%%%%%%%%%%%%%%%
\subsection{Free energy density for nematic elastomers}\label{ssec:FreeEnergy}

 A widely accepted theory for the free energy or entropic elasticity of nematic elastomers is due to Bladon {\it et al.} \cite{btw_pre_93} (see also Warner and Terentjev \cite{wt_lceboox_03}).  As with the classical neo-Hookean model for rubbery solids, this formulation emerges from the statistics of polymer chain conformations, with the caveat being that the distribution properly account for nematic anisotropy associated with the liquid crystal rod-like molecules.  The free energy has the form
\begin{align}\label{eq:We}
W^e(F,n) := \frac{\mu}{2}\begin{cases}
\Tr( F^T \ell_{n}^{-1} F) - 3, &\text{ if } n \in \mathbb{S}^2, \det F = 1 \\
+\infty &\text{ otherwise},
\end{cases}
\end{align}
where $F\in \mathbb{R}^{3\times3}$ is the deformation gradient from the isotropic reference configuration,
\begin{align}
\ell_n := r^{-1/3}( I + (r-1) n \otimes n) \label{eq:nStepLength}
\end{align}
is the step-length tensor that incorporates the nematic anisotropy, $n$ is the director in the current configuration and $r \geq 1$ is an order parameter characterizing the shape changing response due to the local alignment of the liquid crystal molecules.  Notice that incompressibility is imposed since nematic elastomers are polymers which are nearly incompressible,
and $n$ is restricted to be a unit vector.  This energy is frame-indifferent and isotropic.

At high temperatures, thermal fluctuations suppress nematic ordering, leading to an isotropic polymer represented by $r = 1$.  In this case, $\ell_n = I$ and $W^e$ reduces to the incompressible neo-Hookean model.  As the temperature is lowered, the tendency for the liquid crystal molecules to align is increased, leading to a more pronounced average alignment modeled by increasing $r$.  The energy reaches its minimum (i.e., it is zero) when the left Cauchy-Green tensor $FF^T = \ell_n$, reflecting the spontaneous elongation along the director and contraction transverse to it.  Notice that larger $r$ corresponds to greater shape changing zero-energy distortion.  For the analysis herein, we keep the temperature fixed and consider only mechanical boundary conditions, so $r$ is simply a fixed constant greater than or equal to 1.

In the experiments of Kundler and Finkelmann \cite{kf_mrc_95} depicted in Figure \ref{fig:CSKF}, the undeformed state is a monodomain sample with director alignment in the vertical direction.  This monodomain synthesis is quite common for nematic elastomers, and one can adapt the entropic elasticity (\ref{eq:We}) to appropriately account for deformation and stress in these samples by a change of reference configuration.  One simply substitutes
\begin{align}\label{eq:isoMono}
F = F_{m} \ell_{n_0}^{1/2}
\end{align}
into the entropic elasticity (\ref{eq:We}) to obtain the appropriate theory \cite{wt_lceboox_03}.   Here, $\ell_{n_0}^{1/2}$ is the square-root of the step-length tensor (\ref{eq:nStepLength}) with $n_0$ replacing $n$, $n_0$ represents the director prior to deformation (e.g., $n_0 = e_2$ in the example in Figure \ref{fig:CSKF}) and $F_m \in \mathbb{R}^{3\times3}$ represents the deformation gradient from the monodomain state.  Later on, we introduce effective and two dimensional theories for nematic elastomers which are systematically derived from the entropic energy (\ref{eq:We}) (deformed from an isotropic reference state).  For a monodomain sheet in which the initial director $n_0$ is in the plane of the sheet, an analogous substitution to (\ref{eq:isoMono}) properly accounts for the distinction between the isotropic reference state and the monodomain reference state in  the resulting two dimensional theories (we show this elsewhere \cite{p_thesis}).

A second contribution to the energy is the Frank elasticity which reflects the and elastic resistance due to deviation of the director field away from the preferred uniform alignment.  It is  given by (see for example De Gennes and Prost \cite{dgp_book_95})
\begin{equation}
W^n = \frac{1}{2} \kappa_1 (\text{div }n)^2 + \frac{1}{2} \kappa_2 (n \cdot \text{curl }n)^2 + 
\frac{1}{2} \kappa_3 (n \times \text{curl }n)^2
\end{equation}
where all the derivatives are spatial.  Since the moduli $\kappa_1, \kappa_2, \kappa_3$ are close to each other, we can approximate this energy (and bound it from above and below as)
\begin{equation}
W^n \approx \frac{1}{2} \kappa |\nabla_y n|^2 
=  \frac{1}{2} \kappa | (\nabla n) F^{-1}|^2 
=  \frac{1}{2} \kappa | (\nabla n) (\text{cof } F)^T|^2 
\end{equation}
where the final equality holds when $\det F =1$ as $\cof F \in \mathbb{R}^{3\times 3}$ denotes the cofactor matrix of $F$.

So, the free energy of a specimen occupying a region $\Omega$ in the isotropic reference state  under the deformation $y$ and with current director field $n$ may be written as
\begin{equation} \label{eq:e}
{\mathcal E}(y,n) = \int_\Omega \left\{ W^e (\nabla y, n) + \frac{1}{2} \kappa | (\nabla n) (\text{cof } \nabla y)^T|^2 \right\} dx. 
\end{equation}

Verwey {\it {\it et al.}\ } \cite{vwt_jphys_96} argued that stripe domains with alternating directors can emerge as minimizers of this energy functional, with transition zones with the length-scale $\sqrt{\kappa/\mu}$.  We note that this length-scale is on the order of $1-50 nm$, much smaller than the typical thickness of a film which is of the order $10-200 \mu m$ \cite{wt_lceboox_03}.

%%%%%%%%%%%%%%%%%%%%%%%%%%%%%%%%%%%%%%%%%
%%%%%%%%%%%%%%%%%%%%%%%%%%%%%%%%%%%%%%%%%
\subsection{A macroscopic three dimensional description via relaxation}\label{ssec:3Dsec}

\begin{figure}
\centering
\includegraphics[width = 6.0 in]{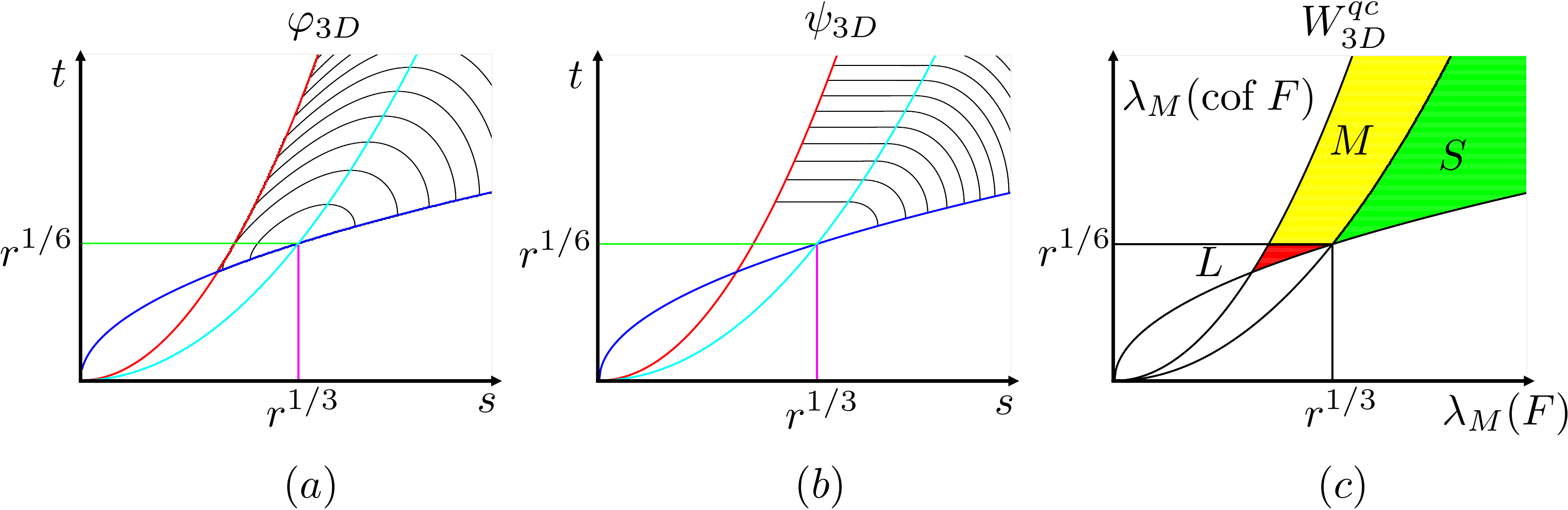}
\caption{Macroscopic three-dimensional energy of nematic elastomers following DeSimone and Dolzmann \cite{dd_arma_02}.  (a) Contour plots of the function $\varphi_{3D}$ that describes the entropic elastic energy $W_{3D}$,  (b) Contour plots of the function $\psi_{3D}$ that describes the relaxed elastic energy $W_{3D}^{qc}$ (i.e., one that implicitly accounts for microstructure and (c) Identification of the regions $L$, $M$ and $S$. }
\label{fig:3DEnergies}
\end{figure}
 
In this section, we recall the results of DeSimone and Dolzmann \cite{dd_arma_02} concerning the macroscopic behavior of nematic elastomers.
Since the Frank energy is small, we can neglect it while studying the behavior of specimens that are large compared to $\sqrt{\kappa/\mu}$.  Thus, we can define a purely mechanical energy density by minimizing out the effect of the director $n$ in $W^e$.  This is given by
\begin{align}\label{eq:W3D}
W_{3D}(F) := \inf_{n \in \mathbb{S}^2} W^e(F,n) = \begin{cases}
\varphi_{3D}(\lambda_M(F), \lambda_M(\cof F))&\text{ if } \det F = 1\\
+\infty & \text{ otherwise},
\end{cases}
\end{align}
where
\begin{align}
\varphi_{3D}(s,t) := \frac{\mu}{2} \left(r^{1/3} \left(\frac{s^2}{r} + \frac{t^2}{s^2} + \frac{1}{t^2}\right) - 3 \right),
\end{align}
$\lambda_M(F)$ is the maximum singular value of $F$ (i.e., the largest principal stretch or the square-root of the maximum eigenvalue of $F^TF$ and $F F^T$) and $\lambda_M(\cof F)$ is the maximum singular value of  $\cof F \in \mathbb{R}^{3\times3}$, (it is easy to show that the this is also equal to the product of the largest two principal values of $F$).  
%Physically, a director $n^{\ast}$ which achieves the minimization above coincides with any direction of maximum stretch for the deformation $F$, i.e.,
%\begin{align}\label{eq:lambdaM}
%W_{3D}(F) = W^e(F,n^{\ast}) \neq + \infty \quad \Rightarrow \quad \lambda_M(F) = |F^T n^{\ast}|.
%\end{align} 
A contour plot of $\varphi_{3D}$ is given in Figure \ref{fig:3DEnergies}(a).

This energy density is not quasiconvex.  Thus, fine-scale microstructure can drive energy minimization in the variational formulation of the elastic energy with this strain energy density, and this leads to a possible non-existence of minimizers.  We account for this by replacing $W_{3D}$ with its relaxation.  Mathematically, this is the quasiconvex envelope of $W_{3D}$ (see Dacorogna \cite{dac_book_07}), 
\begin{align}\label{eq:W3Dqc1}
W_{3D}^{qc}(F) = \inf \left\{ \fint_{\Omega} W_{3D}(F + \nabla \phi) dx \colon  \phi \in W^{1,\infty}_0(\Omega,\mathbb{R}^3) \right\}
\end{align}
where $W^{1,\infty}_{0}(\Omega,\mathbb{R}^3)$ is the space of Lipschitz continuous functions $\phi \colon \Omega \rightarrow \mathbb{R}^3$ which vanish on the boundary of $\Omega$, and $\fint_{\Omega} = \frac{1}{|\Omega|} \int_{\Omega}$ averages the energy density over $\Omega$.  

DeSimone and Dolzmann \cite{dd_arma_02} computed the analytical expression for $W_{3D}^{qc}$ for $W_{3D}$ in (\ref{eq:W3D}), 
\begin{align}\label{eq:W3DQC}
W_{3D}^{qc}(F) = \begin{cases}
\psi_{3D}( \lambda_M(F), \lambda_M( \cof F)) &\text{ if } \det F= 1\\
+\infty & \text{ otherwise}
\end{cases}
\end{align}
where
\begin{align}
&\psi_{3D}(s,t) := \frac{\mu}{2}\begin{cases}
0 &\text{ if } (s,t) \in L \\
r^{1/3}\left(\frac{2t}{r^{1/2}} + \frac{1}{t^2} \right)- 3 &\text{ if } (s,t) \in M \\
r^{1/3}\left(\frac{s^2}{r} + \frac{t^2}{s^2} + \frac{1}{t^2} \right)- 3   &\text{ if }  (s,t) \in S
\end{cases}, \\
&L := \{(s,t) \in \mathbb{R}^{+}\times \mathbb{R}^{+} \colon t \leq s^2,\; t \leq r^{1/6}, \;t \geq s^{1/2} \}, \\
&M := \{(s,t) \in \mathbb{R}^{+} \times \mathbb{R}^{+} \colon t \leq s^2, \;t \geq r^{-1/2} s^2,\; t \geq r^{1/6} \}, \label{eq:setM}\\
&S := \{(s,t) \in \mathbb{R}^{+} \times \mathbb{R}^{+} \colon t \leq r^{-1/2} s^2,\; t \geq s^{1/2} \}.
\end{align}

A contour plot of $W_{3D}^{qc}$ is given in Figure \ref{fig:3DEnergies}(b), and the regions $L$ (of {\it liquid}-like behavior), $M$ (related to stressed {\it microstructure}) and $S$ (of normal {\it solid} behavior) are identified in in Figure \ref{fig:3DEnergies}(c).   Specifically, note that the relaxation $W_{3D}^{qc}$ deviates from the energy density $W_{3D}$ in regions $L$ and $M$ in Figure \ref{fig:3DEnergies}(c).  Importantly, these are the regions where macroscopic deformation can be accommodated by fine-scale oscillations in the director field of a nematic elastomer, resulting in the relaxation having lower energy in region $M$ and zero energy in region $L$.  These features were used by Conti {\it et al.} \cite{cdd_02_pre,cdd_02_jmps} to explain soft elasticity and the complex deformation states in the clamped-stretch experiments of Kundler and Finkelmann \cite{kf_mrc_95} (Figure  \ref{fig:CSKF}) assuming purely planar deformations.

%%%%%%%%%%%%%%%%%%%%%%%%%%%%%%%%%%%%%%%%%
%%%%%%%%%%%%%%%%%%%%%%%%%%%%%%%%%%%%%%%%%
\subsection{Membrane (tension field) theory for nematic elastomer sheets}\label{ssec:EffectiveTheory}

Cesana {\it {\it et al.}\ }\cite{cpk_arma_15} developed a membrane (tension field) theory for nematic elastomers, and we recall their results, emphasizing the ideas that are necessary for our later development.  They consider a  sheet of small thickness $h$ and lateral extent $\omega \subset {\mathbb R}^2$ ($\Omega_h = \omega \times (-h/2, h/2)$), and follow methods in LeDret and Raoult \cite{lr_95_jmpa} and Conti and Dolzmann \cite{cd_06_chap} to study the asymptotic behavior as $h \rightarrow 0$ of the functional $\frac{1}{h} {\mathcal E}^h$ by $\Gamma$-convergence (where ${\mathcal E}^h$ denotes the energy (\ref{eq:e}) of a sheet initially occupying a region $\Omega_h$).  They show that the behavior of very thin sheets is described by the following energy
\begin{equation}
{\mathcal E}_{m}(y) = \int_\omega W_{2D}^{qc} (\tilde \nabla y) d\tilde x
\end{equation}
as the $\Gamma$-limit, a result which is independent of the ratio of $\kappa$ to $h$ (as long as $\kappa \rightarrow 0$ as $h \rightarrow 0$).  
Here, quantities with a tilde denote quantities on the midplane $\omega$ of the membrane, and $y: \omega \to {\mathbb R}^3$ denotes a deformation of the midplane of the membrane.  Thus, $\tilde \nabla y$ maps to  ${\mathbb R}^{3\times 2}$.

\begin{figure}
\centering
\includegraphics[width = 6.0 in]{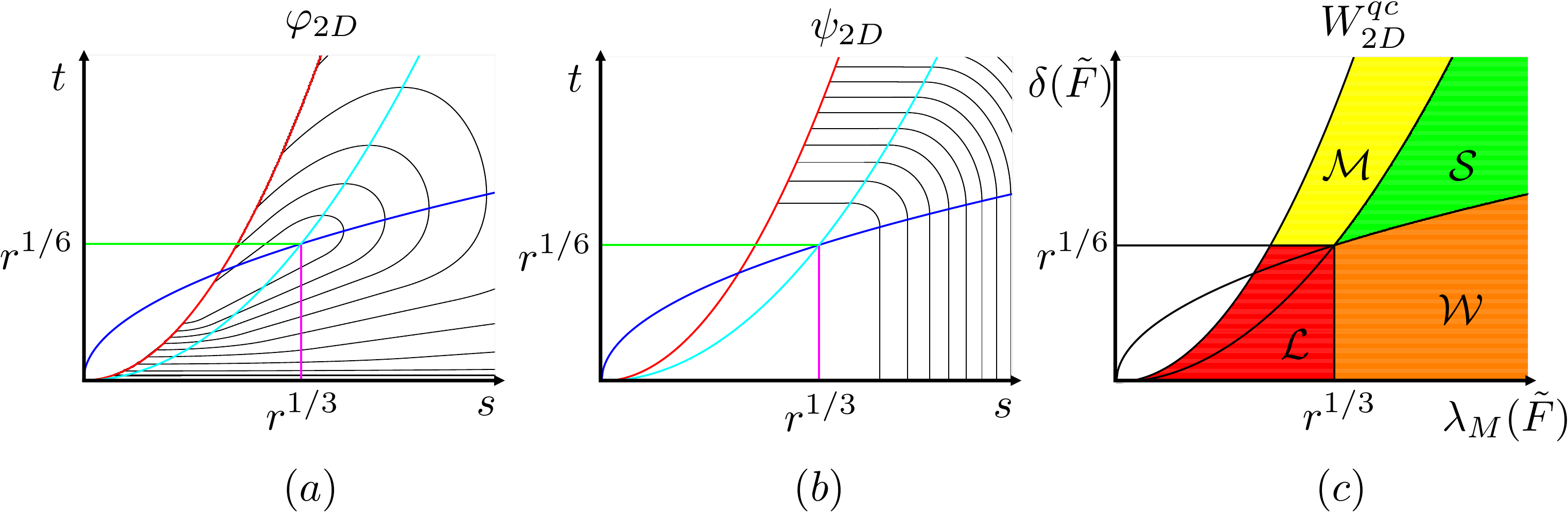}
\caption{Membrane energy density of nematic elastomers following Cesana {\it {\it et al.}\ }\cite{cpk_arma_15}.  (a) Contour plots of the function $\varphi_{2D}$ that describes the plane stress energy $W_{2D}$,  (b) Contour plots of the function $\psi_{3D}$ that describes the relaxed membrane energy $W_{2D}^{qc}$ (i.e., one that implicitly accounts for microstructure and wrinkling) and (c) identification of the regions $\mathcal{L}, \mathcal{M}, \mathcal{W}$ and $\mathcal{S}$.  Microstructure or stripe domains occur in the region ${\mathcal M}$, wrinkling in region ${\mathcal W}$, crumpling and microstructure in region ${\mathcal L}$ and no relaxation in region ${\mathcal S}$.}
\label{fig:2DEnergies}
\end{figure}

The the analytical expression for the membrane energy density $W_{2D}^{qc}$ is obtained from $W_{3D}$ in two steps (the derivation is provided in \cite{cpk_arma_15}).  In the first step, we obtain a two dimensional or plane-stress reduction of $W_{3D}$ by minimizing over the out-of-plane deformation gradient 
\begin{align}\label{eq:W2Ddef}
W_{2D}(\tilde{F}) := \inf_{b \in \mathbb{R}^3} W_{3D}(\tilde{F} |b)
\end{align}
for any planar deformation gradient $\tilde{F} \in \mathbb{R}^{3\times2}$.  This takes the explicit form
\begin{align}\label{eq:W2D}
W_{2D}(\tilde{F}) = \begin{cases} 
\varphi_{2D}(\lambda_M(\tilde{F}), \delta(\tilde{F})) & \text{ if } \rank \tilde{F} = 2  \\
+\infty  &\text{ otherwise},
\end{cases}
\end{align}
where
\begin{align}
&\varphi_{2D}(s,t) := \frac{\mu}{2} \begin{cases}
r^{1/3}\left( \frac{s^2}{r} + \frac{t^2}{s^2} + \frac{1}{t^2}\right) - 3  & \text{ if } (s,t) \in \mathcal{N}_1 \\
 r^{1/3} \left(\frac{t^2}{s^2} + \frac{2s}{r^{1/2} t} \right) - 3 & \text{ if } (s,t) \in \mathcal{N}_2 \\
 r^{1/3} \left( s^2 + \frac{t^2}{s^2} + \frac{1}{r t^2}\right) - 3 & \text{ if } (s,t) \in \mathcal{N}_3 \\
 \end{cases},\\
 &\mathcal{N}_1 := \{(s,t) \in \mathbb{R}^+ \times \mathbb{R}^{+} \colon t \leq s^2, \; t \geq r^{1/2} s^{-1} \}, \\
&\mathcal{N}_2 := \{(s,t) \in \mathbb{R}^+ \times \mathbb{R}^{+} \colon t \leq s^2, \;  r^{-1/2} s^{-1} \leq t \leq r^{1/2} s^{-1} \}, \\
&\mathcal{N}_3 := \{(s,t) \in \mathbb{R}^+  \times \mathbb{R}^{+} \colon t \leq s^2, \; t \leq r^{-1/2} s^{-1} \}.
\end{align}
Above, $\lambda_M(\tilde{F}) := \sup_{e \in \mathbb{S}^2} |\tilde{F}^Te|$ denotes the maximum principal value of $\tilde F$ and $\delta(\tilde F):= |\tilde{F} \tilde{e}_1 \times \tilde{F} \tilde{e}_2|$ denotes the product of the two principal values of $\tilde F$.   Physically, since we are considering the deformations of the plane, the isotropic invariants of the deformation are given by the principal stretch $\lambda_M$ and areal stretch $\delta$.
The contour plot of $\varphi_{2D}$ is shown in Figure \ref{fig:2DEnergies}(a).

As with $W_{3D}$, $W_{2D}$ is not quasiconvex.  In this case, in addition to microstructure in the form of oscillations in nematic orientation relaxing the energy, boundary conditions which induce compressive stresses associated with $W_{2D}$ can be relaxed through out-of-plane wrinkling and crumpling.  We account for this in the second step through the relaxation of $W_{2D}$, 
\begin{align}\label{eq:W2Dqc}
W_{2D}^{qc}(\tilde{F})  := \inf \left\{ \fint_{\omega} W_{2D}(\tilde{F} + \tilde{\nabla} \phi) d\tilde{x} \colon \phi \in W^{1,\infty}_0(\omega,\mathbb{R}^3) \right\} =\psi_{2D} ( \lambda_M(\tilde{F}), \delta(\tilde{F})).
\end{align}
Here,
\begin{align}
&\psi_{2D}(s,t) := \frac{\mu}{2}\begin{cases}
0 &\text{ if } (s,t) \in \mathcal{L} \\
r^{1/3}\left(\frac{2t}{r^{1/2}} + \frac{1}{t^2}\right) - 3 &\text{ if } (s,t) \in \mathcal{M} \\ 
r^{1/3} \left(\frac{s^2}{r} + \frac{2}{s} \right) - 3 &\text{ if } (s,t) \in \mathcal{W} \\
r^{1/3} \left(\frac{s^2}{r} + \frac{t^2}{s^2} + \frac{1}{t^2}\right) - 3   &\text{ if }  (s,t) \in S
\end{cases}, \\
&\mathcal{L} := \{(s,t) \in \mathbb{R}^{+}\times \mathbb{R}^{+} \colon t \leq s^2, t \leq r^{1/6}, s \leq r^{1/3} \}, \\
&\mathcal{M} := \{(s,t) \in \mathbb{R}^{+} \times \mathbb{R}^{+} \colon t \leq s^2, t \geq r^{-1/2} s^2, t \geq r^{1/6} \}, \\ 
&\mathcal{W} := \{(s,t) \in \mathbb{R}^{+} \times \mathbb{R}^{+} \colon t \leq s^{1/2}, s \geq r^{1/3} \}, \label{eq:w} \\
&\mathcal{S} := \{(s,t) \in \mathbb{R}^{+} \times \mathbb{R}^{+} \colon t \leq r^{-1/2} s^2, t \geq s^{1/2} \}.
\end{align}
The contour plot of $\psi_{2D}$ is shown in Figure \ref{fig:2DEnergies}(b), and the various regions\footnote{The notation $\mathcal{S}$ and $\mathcal{L}$ follows DeSimone and Dolzmann \cite{dd_arma_02} in their derivation of the relaxed three dimensional theory.  However, one should avoid the interpretation of {\it liquid}-like behavior for $\mathcal{L}$ and normal {\it solid} behavior for $\mathcal{S}$. $\mathcal{L}$ is associated both to crumpling $-$ which is also exhibited by normal thin solids $-$ as well as liquid-like features due to microstructure.  $\mathcal{S}$ is simply a region without instability.} $\mathcal{L}$ (of zero energy), $\mathcal{M}$ (related to stressed {\it microstructure}), $\mathcal{W}$ (related to {\it wrinkling}) and $\mathcal{S}$ (without instability) are shown in 
Figure \ref{fig:2DEnergies}(c).

It is instructive to look at the stress that results from this theory.  We can obtain the effective Cauchy stress of a nematic elastomer membrane as 
\begin{align}\label{eq:Cauchy}
\sigma^{mem} := (W_{2D}^{qc}),_{\tilde{F}} \tilde{F}^T, \quad \text{i.e.,} \quad (\sigma^{mem})_{ij} := (W_{2D}^{qc}),_{\tilde{F}_{i \alpha}} \tilde{F}_{j \alpha}, \quad \alpha = 1,2, \quad i,j = 1,2,3.
\end{align} 
To compute it explicitly, we use the singular value decomposition theorem to write
\begin{align}\label{eq:SVD}
\tilde{F} = \bar{\lambda}_M g_1 \otimes \tilde{f}_1 + \bar{\lambda}_m g_2 \otimes \tilde{f}_2
\end{align}
for orthonormal vectors $\{\tilde{f}_1, \tilde{f}_2\} \subset \mathbb{R}^2$ and $\{g_1, g_2\} \subset \mathbb{R}^3$ with $\bar{\lambda}_M \geq \bar{\lambda}_m \geq 0$ the singular values of $\tilde{F}$ (so that $\bar{\delta} = \bar{\lambda}_m \bar{\lambda}_M$). We find (again see \cite{cpk_arma_15}),
\begin{align}\label{eq:Cauchy1}
\sigma^{mem} = \mu r^{1/3} \begin{cases}
0 & \text{ if } (\bar{\lambda}_M,\bar{\delta}) \in \mathcal{L} \\
\left(\frac{\bar{\delta}}{r^{1/2}} - \frac{1}{\bar{\delta}^2} \right) (g_1 \otimes g_1 + g_2 \otimes g_2)  & \text{ if } (\bar{\lambda}_M ,\bar{\delta}) \in \mathcal{M} \\
\left(\frac{\bar{\lambda}_M^2}{r} - \frac{1}{\bar{\lambda}_M} \right) g_1 \otimes g_1 & \text{ if } (\bar{\lambda}_M, \bar{\delta}) \in \mathcal{W} \\
\left( \frac{\bar{\lambda}_M^2}{r} - \frac{1}{\bar{\delta}^2}\right) g_1 \otimes g_1 \left(\frac{\bar{\delta}^2}{\bar{\lambda}_M^2} - \frac{1}{\bar{\delta}^2} \right) g_2 \otimes g_2 & \text{ if } (\bar{\lambda}_M, \bar{\delta}) \in \mathcal{S}. 
\end{cases}
\end{align}

This formula highlights some striking features the membrane theory for nematic elastomers.  For one, the membrane is always in a state of plane stress in the tangent plane.  Secondly, the principal stresses (i.e., the eigenvalues of $\sigma^{mem}$) are always non-negative. Therefore, these membranes cannot sustain compressive stress.  Further, the stress is zero in region $\mathcal{L}$ where crumpling and microstrucutre relax the energy to zero, uniaxial tension in $\mathcal{W}$ where tension wrinkling relaxes the energy, equi-biaxial tension in $\mathcal{M}$ where microstructure relaxes the energy and biaxial tension in $\mathcal{S}$ where no fine-scale features emerge to relax the energy.  

For perspective, consider the special case $r = 1$ when this theory reduces to that of the neo-Hookean elastic membrane (recall $W^e$ in (\ref{eq:We}) simplifies to the incompressible neo-Hookean energy density in this case).  The region $\mathcal{M}$ now disappears and we are left with regions $\mathcal{L}, \mathcal{W}$ and $\mathcal{S}$ with zero, uniaxial tension and biaxial tension respectively.  This is the tension field theory originally proposed by Mansfield \cite{m_70_prsa} and later obtained systematically from three dimensional elasticity by Pipkin \cite{p_86_jam} and expanded upon by Pipkin and Steigmann in \cite{s_90_prsa,sp_89_qjmam}.  In essence, the theory for nematic elastomer membranes with $r > 1$ generalizes the tension field theory for isotropic membranes to account for nematic anisotropy. 

A remarkable feature of nematic elastomers is the additional region $\mathcal{M}$ where the state of stress is equi-biaxial tension.  This is true for a large range of unequal principal stretches $(\bar{\lambda}_M, \bar{\lambda}_m)$.  In other words, {\it a nematic elastomer membrane can have shear strain without shear stress} in a certain range.

%%%%%%%%%%%%%%%%%%%%%%%%%%%%%%%%%%%%%%%%%
%%%%%%%%%%%%%%%%%%%%%%%%%%%%%%%%%%%%%%%%%
\section{Koiter theory for nematic elastomer sheets} \label{sec:k}

The membrane energy described in the previous section relaxes over both microstructure and wrinkling.  While this can provide insights, it is not sufficient for our present purpose of understanding the formation of wrinkles.  We seek, instead, a theory that can explicitly describe the wrinkles.  However, as mentioned in the introduction, nematic elastomers can also form microstructure, and for the sheets relevant to the clamped stretch experiments, this microstructure is very fine compared to the wrinkles since $\sqrt{\kappa/\mu} \ll h$.  Therefore, we seek a theory that both relaxes over microstructure and resolves wrinkles. This motivates a Koiter-type theory for nemetic elastomer sheets.  In this section, we first describe the theory and subsequently justify it.

\subsection{Koiter theory for nematic elastomer sheets} 
We consider a nematic sheet with midplane $\omega \subset \mathbb{R}^2$ in the isotropic reference state.  For a midplane deformation $y \colon \omega \rightarrow \mathbb{R}^3$ of the sheet,  we take the Koiter elastic strain energy to be
\begin{align}\label{eq:Koiter}
\mathcal{E}_{K}^h(y) := h \int_{\omega} \left( W_{ps}(\tilde{\nabla} y) + \frac{\mu r^{1/3} h^2}{6} | \II_y|^2 \right) d\tilde{x}.
\end{align}
Here, the first term is term proportional to the thickness and describes the in-plane deformation accounting for the microstructure.  It is given by the nematic plane-stress energy density\\
\begin{align}\label{eq:WpsDef}
& W_{ps} (\tilde{F}) := \varphi_{ps}( \lambda_M(\tilde{F}), \delta(\tilde{F})), \\
&
\varphi_{ps}(s,t) := \begin{cases}
\varphi_{2D} ( s,t) &\text{ if } (s,t) \in \mathcal{C}\\
\psi_{2D} ( s,t) &\text{ if } (s,t) \in \mathcal{L}_m \cup \mathcal{M} \cup \mathcal{S}
\end{cases}, \\
&\mathcal{L}_m := \{(s,t) \in \mathbb{R}^{+}\times \mathbb{R}^{+} \colon t \leq s^2, r^{1/6} \ge t \ge r^{-1/6}s \}, \\
&\mathcal{M} := \{(s,t) \in \mathbb{R}^{+} \times \mathbb{R}^{+} \colon t \leq s^2, t \geq r^{-1/2} s^2, t \geq r^{1/6} \}, \\ 
&\mathcal{S} := \{(s,t) \in \mathbb{R}^{+} \times \mathbb{R}^{+} \colon t \leq r^{-1/2} s^2, t \geq s^{1/2} \}, \\
&\mathcal{C} := \{(s,t) \in \mathbb{R}^{+} \times \mathbb{R}^{+} \colon t \leq s^{1/2},  t \le r^{-1/6}s, t \le s^{1/2}
 \} .
\end{align}
Further, the second term is proportional to the cube of the thickness and describes the bending.  It penalizes the second fundamental form of the deformed membrane,   
\begin{align}\label{eq:secFundForm}
\II_y := (\tilde{\nabla} y)^T \tilde{\nabla} \nu_y, \quad \nu_y := \frac{\partial_1 y \times \partial_2 y}{|\partial_1 y \times \partial_2 y|}.
\end{align}
Note that $\nu_y$ is the surface normal to the deformed midplane of the sheet.

\begin{figure}
\centering
\includegraphics[width = 5 in]{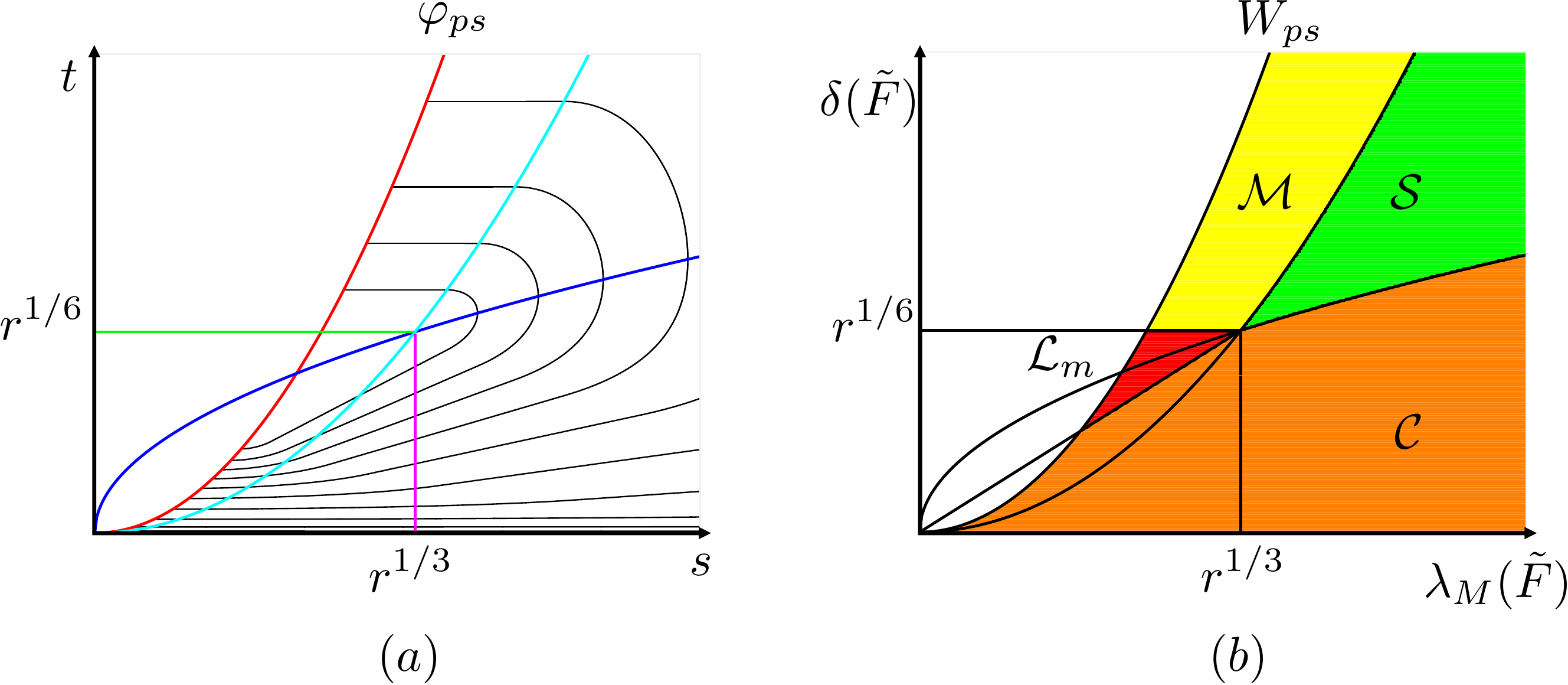}
\caption{The plane-stress energy density $W_{ps}$ for nematics.  (a) Contour plots of the function $\varphi_{ps}$ that describes the plane stress energy $W_{ps}$,  (b) identification of the regions $\mathcal{L}_m, \mathcal{M}, \mathcal{C}$ and $\mathcal{S}$.}
\label{fig:KEnergy}
\end{figure}

Before proceeding with the derivation of this theory (Section \ref{ssec:OutlineK}-\ref{ssec:TensEnergy}), we digress to comment on some of its more striking features:

First, the plane-stress energy density $W_{ps}$ has a revealing interpretation.  We note that it is actually derived from $W_{3D}^{qc}$, an energy density for which microstructure is already completely relaxed. (This is made precise in the coming sections: see, for instance, (\ref{eq:DeriveStretch})). Therefore, {\it microstructure is completely relaxed in this plane stress energy density}.  Particularly the region $\mathcal{L}_m$ (of {\it liquid}-like behavior) corresponds to the relaxation of microstructure and a state of zero energy and stress, and regions $\mathcal{M}$ and $\mathcal{S}$ are exactly as in the relaxed membrane theory.  

Next, note that wrinkling and crumpling are not relaxed in this theory as region $\mathcal{C}$ (which coincides with unstable regions of $W_{2D}$, see (\ref{eq:stretchEnergy})) is a region of {\it compressive} stresses in this plane-stress energy density.  These instabilities are instead regularized by bending (i.e., the second term in (\ref{eq:Koiter})).  

Next, turning to bending, note that modulus of bending here is $\mu r^{1/3} h^3/3$ (two times the constant in front of the bending term in (\ref{eq:Koiter}) as the modulus is by convention proportional to the stress, not the energy).  This can be rewritten in terms of the Young's Modulus $E$ as $Eh^3 r^{1/3}/9$ since nematic elastomers are incompressible with a Poisson's ratio of $1/2$.  In setting $r = 1$, corresponding to an incompressible neo-Hookean sheet, this reduces to $Eh^3/9=D$ for $D$ the bending modulus of an incompressible and isotropic plate of initial thickness $h$ in classical plate theory.  Thus, this bending modulus is  consistent with classical plate theory.  

As a final comment related to bending, note that changes in thickness associated with finite-deformation are properly accounted for even though it appears that the modulus only depends on the initial thickness $h$. This is a consequence of the fact that the second fundamental form (\ref{eq:secFundForm}) is computed with respect to the reference configuration.  For instance, imagine first deforming the specimen from the stress-free isotropic configuration $\omega$ to a stretched configuration $\omega_{\lambda} := (\lambda \tilde{e}_1 \otimes \tilde{e_1} + \lambda^{-1/2} \tilde{e}_2 \otimes \tilde{e}_2) \omega$ (i.e., accounting for the natural lateral contraction due to incompressibility). Then imagine measuring bending transverse to this stretch while taking $\omega_{\lambda}$ as the reference configuration.  In this scenario, the modulus of bending is  $\mu r^{1/3} h_{\lambda}^3/3$ for $h_{\lambda} = h \lambda^{-1/2}$.  Thus as expected, the modulus depends on the deformed thickness associated with the natural transverse contraction of an incompressible sheet under stretch.  

\subsection{Outline of the derivation of the Koiter Theory}\label{ssec:OutlineK}
We now sketch the derivation of the Koiter theory: First, we show how the plane stress energy generically emerges from the relaxed three dimensional free energy for any given midplane deformation.  Then, we outline the argument for the bending term, which accounts for the energy due to tension wrinkling. 

Let us first consider a nematic elastomer sheet with finite but small thickness $h \ll 1$ occupying the region $\Omega_h := \omega \times (-h/2,h/2)$ where $\omega \subset \mathbb{R}^2$ is the midplane in its isotropic reference state.  We assume that the microstructure is fully-relaxed so that a deformation of this region $y^h \colon \Omega_h \rightarrow \mathbb{R}^3$ is subject to a strain energy 
\begin{align}\label{eq:E3DQC}
\mathcal{E}_{3D}^h(y^h) := \int_{\Omega_h} W_{3D}^{qc}(\nabla y^h) dx.
\end{align}
Since we are dealing with a thin sheet, we make the ansatz that any low energy deformation $y^h$ with corresponding midplane deformation $y(\tilde{x}) := y^h(\tilde{x},0)$ behaves to leading order in $x_3$ as  
\begin{align} \label{eq:ext}
y^h(\tilde x, x_3) \approx y (\tilde x) + x_3 b(\tilde{x})
\end{align}
for some vector $b\colon \omega \to {\mathbb R}^3$.  (Note, this approximation can be obtained by Taylor expanding $y^h$ in $x_3$ about $x_3 = 0$, observing that $x_3 \in (-h/2,h/2)$ is small, and neglecting terms of $O(x_3^2)$ and above). Thus, assuming the $y^h$ is incompressible (i.e., $\det \nabla y^h = 1$ on $\Omega_h$), we find, in substituting (\ref{eq:ext}) into (\ref{eq:E3DQC}), that the energy for these generic deformations scales with the thickness as 
\begin{align} \label{eq:exp}
\mathcal{E}^h_{3D}(y^h) = h \int_{\omega} W_{3D}^{qc}(\tilde{\nabla} y|b) d \tilde x+ O(h^2)
\end{align}
due to the Lipschitz continuity of $W_{3D}^{qc}$ (on matrices with $\det F = 1$).  Now, in fixing the midplane deformation $y \colon \omega \rightarrow \mathbb{R}^3$, we observe that for sufficiently thin sheets, setting 
\begin{align}
b(\tilde{x}) = \arg\min_{\mathbb{R}^3} W_{3D}^{qc}(\tilde{\nabla} y(\tilde{x}) | \cdot ), \quad \tilde{x} \in \omega
\end{align}
approximately minimizes the relaxed three dimensional strain energy.   Therefore, we expect the energy for a generic midplane deformation of a sufficiently thin nematic elastomer sheet to behave as 
\begin{align}\label{eq:EnergyScale}
\mathcal{E}_{3D}^h(y^h) = h \int_{\omega} W_{ps} (\tilde{\nabla} y) d\tilde{x}  + O(h^2)
\end{align}
since actually 
\begin{align}\label{eq:DeriveStretch}
W_{ps}(\tilde{F}) = (W_{3D}^{qc})_{2D}(\tilde{F}) := \inf_{\mathbb{R}^3} W_{3D}^{qc}(\tilde{F}| \cdot)
\end{align}
(for this equality, see (\ref{eq:stretchEnergy}) and the discussion therein).  Thus, the first term in the energy (\ref{eq:EnergyScale}) coincides with the first term in the Koiter theory (\ref{eq:Koiter}).

We will show that $W_{ps}$ is not quasiconvex\footnote{Specifically, we show that $W_{ps}=W_{2D}^{qc}$ everywhere except the region ${\mathcal C}$ which crucially includes ${\mathcal W}$.}.  This means that there are certain midplane deformations $y$ which can be approximated by deformations $y_k \approx y$ which lower the energy through fine scale instabilities.  Mathematically, these midplane deformations have the property
\begin{equation}
\begin{aligned}\label{eq:lscfailure}
&\lim_{k \rightarrow \infty}  \int_{\omega} W_{ps}(\tilde{\nabla} y_k) d \tilde{x} = \int_{\omega} W_{2D}^{qc}(\tilde{\nabla} y) d \tilde{x} < \int_{\omega} W_{ps}(\tilde{\nabla} y) d\tilde{x}\\
&\quad \text{ for some } \quad  y_k \rightharpoonup y \text{ in } W^{1,2}(\omega,\mathbb{R}^3) \quad \text{ as } k \rightarrow \infty.
\end{aligned} 
\end{equation} 
In particular, this can happen for taut membranes, our topic of particular interest, when the deformation gradient $\tilde{F}_0 \in \mathbb{R}^{3\times2}$ corresponds to a point in the region $\mathcal{W}$.  We will show that the instabilities (in the sense of (\ref{eq:lscfailure}) with $y = \tilde{F}_0 \tilde{x}$) necessarily correspond to deformations of the form
\begin{align}\label{eq:ykDef}
y_k(\tilde{x}) = \bar{\lambda}_M x_1 e_1 + \bar{\lambda}_M^{-1/2}\gamma_k(x_2), \quad \tilde{x} \in \omega
\end{align}
for $\bar{\lambda}_M := \lambda_M(\tilde{F}_0)$ and some one-dimensional {\it wrinkled} curve $\gamma_k$ satisfying pointwise the constraints
\begin{align}\label{eq:gammak}
\gamma_k \cdot e_1 = 0, \quad |\gamma_k' | =1 
\end{align}
up to a rigid body rotation and/or a change in coordinates frame (dictated by $\tilde{F}_0$).  Further, the wrinkled curves $\gamma_k$ become finer and finer as $k$ increases so that in the limit the wrinkles are infinitely fine.  

However on this last point, we recognize this fineness as an artifact of the fact that we have not taken into account terms in (\ref{eq:EnergyScale}) which are higher order in the thickness $h$, and that bending (at a scale $h^3$) should emerge to resist or {\it regularize} such infinitely fine wrinkles.  Therefore, we extend the 
midplane deformations (\ref{eq:ykDef}) to the entire sheet $\Omega_h$ while respecting incompressiblity  (i.e., through an incompressible deformation $y_k^h \colon \Omega_h \rightarrow \mathbb{R}^3$ with $y_k^h(\tilde{x}, 0) = y_k(\tilde{x})$), and we find that bending emerges as the energy is characterized by 
\begin{align}
\mathcal{E}^h_{3D}(y_k^h) \approx \mathcal{E}_{K}^h(y_k).
\end{align}
Consequently, the Koiter theory (\ref{eq:Koiter}) appropriately captures the wrinkling behavior of taut nematic elastomer sheets.

\begin{figure}
\centering
\begin{subfigure}{\linewidth}
\centering
\includegraphics[width = 5.0 in]{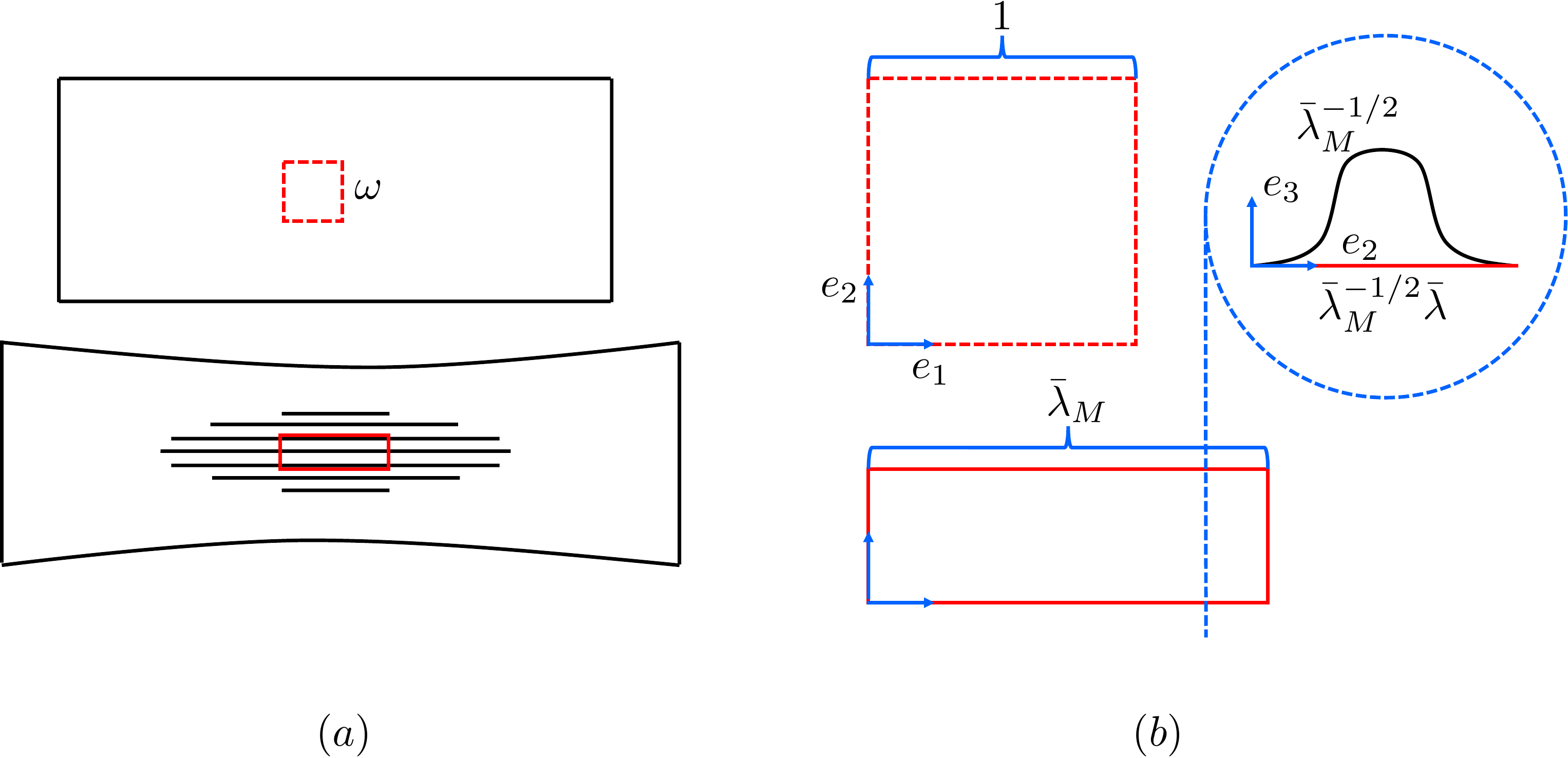}
\end{subfigure}
\begin{subfigure}{\linewidth}
\centering
\includegraphics[width = 5.0 in]{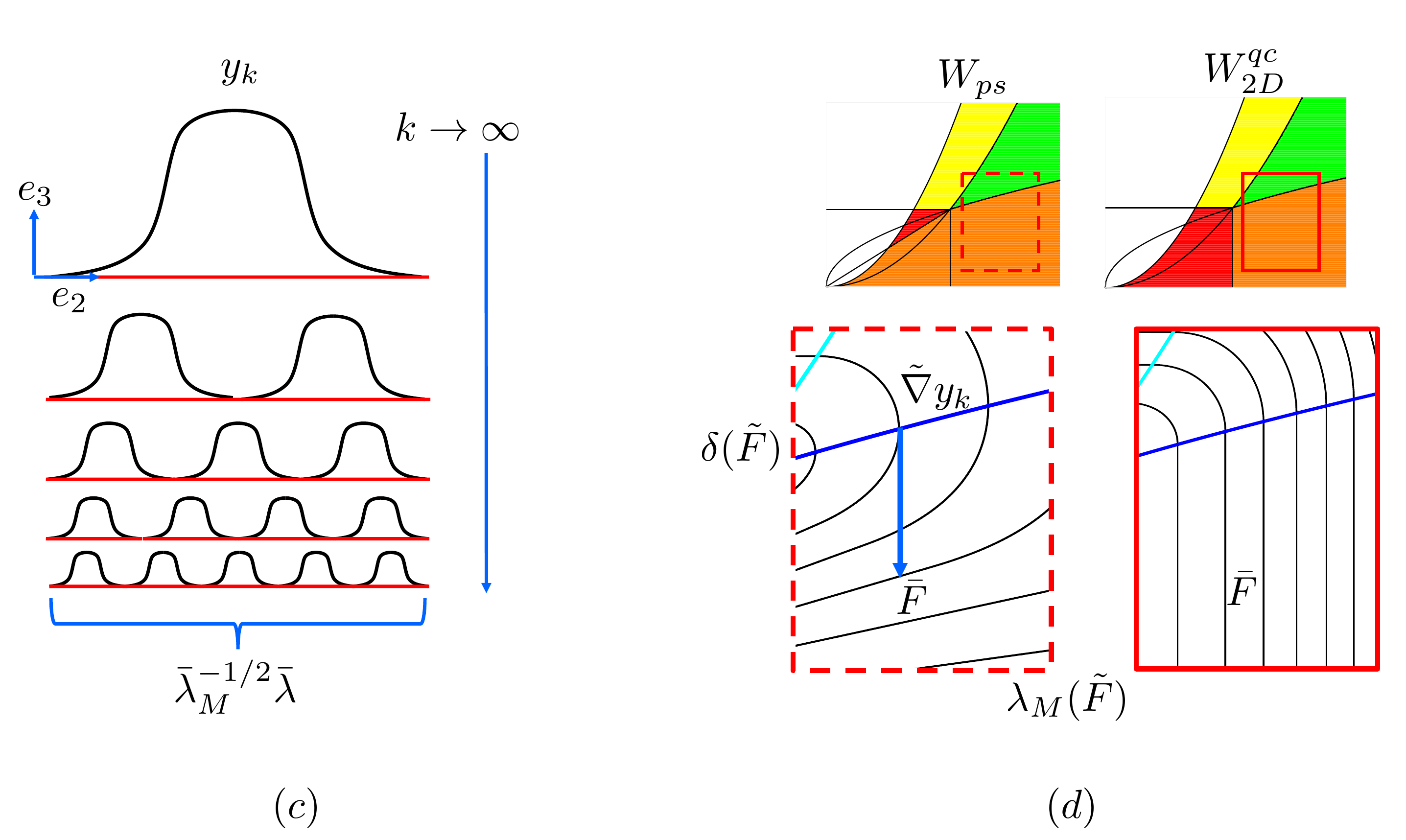}
\end{subfigure}
\caption{A few ideas related to the development of the Koiter theory for nematic elastomer sheets.  (a) A schematic representation of the clamped-stretch experiment.  (b)  The mesoscale deformation of a square shown outlined in red from part (a), and the construction of one wrinkle.  (c) A cross-section through the line marked in part (b) showing the canonical construction of a sequence of finer and finer wrinkles converging weakly to a flat deformation.   (d) The overall deformation gradient $\bar{F}$  corresponding to the point marked is obtained by deformations whose gradient takes the pointwise values corresponding to the point marked $\tilde{\nabla} y_k$.   This relaxes the energy $W_{ps}$ to the membrane energy $W_{2D}^{qc}$.}
\label{fig:WrinkCS}
\end{figure}

\subsection{Characterization of midplane deformations during tension wrinkling}\label{ssec:WrinkleChar}

To develop the Koiter theory, and in particular the bending term which emerges to regularize wrinkling, we first characterize tension wrinkling deformations in the asymptotic context of (\ref{eq:lscfailure}).  We sketch a few ideas related to this development in Figure \ref{fig:WrinkCS}.

We begin with the observation that 
\begin{align}\label{eq:stretchEnergy}
W_{ps}(\tilde{F}) = (W_{3D}^{qc})_{2D}(\tilde{F}) =\begin{cases}
W_{2D}(\tilde{F}) &\text{ if } (\lambda_M(\tilde{F}), \delta(\tilde{F})) \in \mathcal{C} \\
W_{2D}^{qc}(\tilde{F}) & \text{ if }   (\lambda_M(\tilde{F}), \delta(\tilde{F})) \in \mathcal{L}_{m} \cup \mathcal{M}\cup \mathcal{S}.
\end{cases}
\end{align}
This is an equivalent representation to (\ref{eq:WpsDef}) and the result of a long calculation proved in Proposition \ref{W3Dqc2D} in Appendix \ref{sec:KoiterAppend1}.  
In other words, $W_{ps}$ differs from the relaxed membrane energy $W_{2D}^{qc}$ only in the region $\mathcal{C}$ that consists of the region  $\mathcal{W}$ where wrinkling instability occurs and (a large portion of) the region $\mathcal{L}$ where one has crumpling or compression in both directions.  Therefore, all instabilities of 
$W_{ps}$ (i.e., (\ref{eq:lscfailure})) are related to mesoscale deformation gradients in the region $\mathcal{C}$.

We are interested in characterizing tension wrinkling, i.e., the instabilities due to mesoscale deformation gradients in $\mathcal{W}$.  Hence, we focus on these.  Specifically, we take a (region of a) membrane whose midplane is a unit square (i.e., $\omega = (0,1)^2$ the unit square) and subject it to a mesoscale deformation gradient $\tilde{F}_0 \in \mathbb{R}^{3\times3}$ in the wrinkling region, i.e.,  having $(\lambda_M(\tilde{F}_0), \delta(\tilde{F}_0)) \in \mathcal{W}$.  Since it develops instability, we can find a sequences of midplane deformations $y_k$ that satisfy 
\begin{align}\label{eq:sequenceImport}
y_k \rightharpoonup \tilde{F}_0 \tilde{x} \quad \text{ in } W^{1,2}(\omega, \mathbb{R}^3), \quad \int_{\omega} W_{ps}(\tilde{\nabla} y_k) d\tilde{x} \rightarrow  W_{2D}^{qc}(\tilde{F}_0) \quad \text{ as } k \rightarrow \infty.
\end{align} 
We show that such sequences are necessarily of the form (\ref{eq:ykDef}) up to a rigid body rotation and/or a change in coordinates frame.

For this characterization, we note that due to the frame invariance and isotropy of $W^{qc}_{2D}$ and $W_{ps}$ -- the fact that $W^{qc}_{2D}(Q\tilde{F}\tilde{R}) = W^{qc}_{2D}(\tilde{F})$ for all $\tilde{F} \in \mathbb{R}^{3\times2}$, $Q \in SO(3)$ and $\tilde{R} \in SO(2)$ (and similarly for $W_{ps}$) -- we only need to consider the case that $\tilde{F}_0$ is diagonal, i.e., 
\begin{align}\label{eq:barFWrink}
\tilde{F}_0 = \left(\begin{array}{cc} \bar{\lambda}_M & 0 \\ 0 & \bar{\delta}/ \bar{\lambda}_M \\ 0 & 0 \end{array}\right), \quad (\bar{\lambda}_M, \bar{\delta}) \in \mathcal{W}.
\end{align}
In this setting, $\bar{\lambda}_M > r^{1/3}$ and $\bar{\delta} < \bar{\lambda}_M^{1/2}$ (recall the definition of $\mathcal{W}$ in (\ref{eq:w})).  For future use we set 
\begin{align}\label{eq:thetaStuff}
\bar{\delta} = \bar{\lambda}_M^{1/2} \bar{\lambda}, \quad \bar{\lambda} \in (0,1).
\end{align}

Now, it follows from a result of Cesana {\it {\it et al.}\ }\cite{cpk_arma_15} (Theorem 6.1, Part 2), that sequences that satisfy (\ref{eq:sequenceImport}) have the property that 
\begin{align}\label{eq:nablaykCompact}
\tilde \nabla y_k(\tilde{x}) \in K_{\mathcal{W}}
\end{align}
on $\omega$ except perhaps in boundary layers\footnote{Their result shows that the Young measure is supported on $K_{\mathcal{W}}$ which implies 
(\ref{eq:nablaykCompact}) in an appropriate sense.}. Here, 
\begin{equation}
\begin{aligned}
K_{\mathcal{W}} := \left\{ \tilde{F} \in \mathbb{R}^{3\times2} \colon (\lambda_{M}(\tilde{F}), \delta(\tilde{F})) = (\bar{\lambda}_M, \bar{\lambda}_M^{1/2}), \quad \tilde{F} \tilde{e}_1 = \bar{\lambda}_M e_1 \right\}. 
\end{aligned} 
\end{equation}

We note that the condition (\ref{eq:nablaykCompact}) is exactly equivalent to the characterization in (\ref{eq:ykDef}) and (\ref{eq:gammak}). First, it follows from the definition of $K_{\mathcal{W}}$ that $\partial_1 y_k = \bar{\lambda}_M e_1$.  Integrating this yields $y_k(\tilde{x}) = \bar{\lambda}_M x_1 e_1 + \bar{\lambda}_M^{-1/2}\gamma_k(x_2)$ for some $\gamma_k \colon (0,1) \rightarrow \mathbb{R}^3$.  Second, it follows from definition of $K_{\mathcal{W}}$ that $\lambda_M(\tilde{\nabla} y_k) = \bar{\lambda}_M$, and this leads necessarily to the conclusion $\gamma_k' \cdot e_1 = 0$. Hence, $\delta(\tilde{\nabla} y_k) = \bar{\lambda}_M^{1/2} |\gamma_k'|$.  Finally, the requirement from $K_{\mathcal{W}}$ that $\delta(\tilde{\nabla}y_k) = \bar{\lambda}_M^{1/2}$ implies that $|\gamma_k'| = 1$.  Thus, we obtain 
(\ref{eq:ykDef}) and (\ref{eq:gammak}) as asserted.

Rather strikingly, these deformations correspond to pure uniaxial tension where the membrane stretches along its length by a factor $\bar{\lambda}_M$ and deforms transversely out-of-plane to preserve its natural width and avoid compression.   Moreover, the convergence (\ref{eq:sequenceImport}) which leads to the relaxation is assured if and only if 
\begin{align}\label{eq:gammakweak}
\gamma_k \rightharpoonup \bar{\lambda} x_2e_2 \quad \text{ in } W^{1,2}((0,1),\mathbb{R}^3) \quad \text{ as } k \rightarrow \infty.
\end{align}
This combined with the constraint (\ref{eq:gammak}) implies that the curves $\gamma_k$ must oscillate out-of-plane on a fine-scale as $k \rightarrow \infty$ to relax the energy.

%
%
%
%Indeed, any $y_k$ which satisfies  (\ref{eq:nablaykCompact}) is equivalently parameterized by
%\begin{align}\label{eq:ykDef}
%y_k(\tilde{x}) = \bar{\lambda}_M x_1 e_1 + \bar{\lambda}_M^{-1/2}\gamma_k(x_2), \quad \tilde{x} \in \omega
%\end{align}
%for some curve $\gamma_k \colon (0,1) \rightarrow \mathbb{R}^3$ satisfying 
%\begin{align}\label{eq:gammak}
%\gamma_k \cdot e_1 = 0, \quad |\gamma_k' | =1 \quad \text{ on } (0,1).  
%\end{align}
%Given this characterization, the convergence (\ref{eq:weakConvyk}) is assured if and only if 
%\begin{align}\label{eq:gammakweak}
%\gamma_k \rightharpoonup \sin(\theta) x_2 \quad \text{ in } W^{1,2}((0,1),\mathbb{R}^3) \quad \text{ as } k \rightarrow \infty.
%\end{align}
%As a result, $y_k$ necessarily corresponds to fine-scale tension wrinkling.  To see this, we note that each curve $\gamma_k$ has a total arclength of $1$ given its derivative condition in (\ref{eq:gammak}).  Simultaneously, the convergence (\ref{eq:gammakweak}) implies that the curves $\gamma_k$ must approximate pointwise (as $k \rightarrow \infty$) a flat curve which has total length of $\sin(\theta) < 1$.  To do this, $\gamma_k$ must oscillate out-of-plane on a fine-scale as $k \rightarrow \infty$ akin to the example in Figure \ref{fig:WrinkCS}(c).  

It is useful to illustrate these concepts through a canonical construction of tension-wrinkling curves which have the properties (\ref{eq:gammak}) and (\ref{eq:gammakweak}): Consider any smooth curve $\gamma \colon (0,1) \rightarrow \mathbb{R}^3$ in the $\{ e_2, e_3\}$ plane describing a single wrinkle which has a total length of $1$, is given by a regular parameterization (i.e., $\gamma' \neq 0$ on $(0,1)$), and satisfies $\gamma(0) = 0$, $\gamma(1) = \bar{\lambda}e_2$, $\gamma \cdot e_3 = 0$ in a neighborhood of $0$ and $1$.  The example curve $\bar{\lambda}_M^{-1/2} \gamma$ in Figure \ref{fig:WrinkCS}(b) has each of these properties for $\gamma$.  Importantly, any curve $\gamma$ with these properties can be reparameterized by its arc-length yielding a new parametrization $\gamma_0 \colon (0,1) \rightarrow \mathbb{R}^3$ corresponding to the same curve which satisfies (\ref{eq:gammak}).  For a curve $\gamma_k$ of $k$ wrinkles which satisfies (\ref{eq:gammak}), we simply extend $\gamma_0$ periodically to all of $\mathbb{R}$ so that $\gamma_0(x_2 + 1) = \gamma_0(x_2)$ for $x_2 \in \mathbb{R}$ and set 
\begin{align}
\gamma_k(x_2) = k^{-1} \gamma_0(k x_2), \quad x_2 \in (0,1).
\end{align}
This rescaling is akin to the depiction in Figure \ref{fig:WrinkCS}(c) where the curve in (b) is rescaled in precisely the same manner to obtain a curve of $k$ wrinkles of wavelength $1/k$.  Moreover, any curve $\gamma_k$ constructed in this manner satisfies (\ref{eq:gammakweak}) and obeys the estimates
\begin{equation}
\label{eq:estimateEx}
\begin{aligned}
\|\gamma_k - \bar{\lambda}x_2 e_2\|_{L^{\infty}} = k^{-1} \| \gamma_0 - \bar{\lambda} x_2 e_2\|_{L^{\infty}} \\
\|\gamma_k''\|_{L^{\infty}} = k \|\gamma_0''\|_{L^{\infty}}, \quad \|\gamma_k'''\|_{L^{\infty}} = k^2\|\gamma_0'''\|_{L^{\infty}}.
\end{aligned}
\end{equation}
Thus, we see that as $k$ increases, the curve $\gamma_k$ converges to the flat line $\bar{\lambda}x_2 e_2$, but the fine-scale oscillations needed to obtain this convergence result in amplified curvature and higher derivatives (as seen by the dependence on $k$ in the derivatives).  

While the canonical construction is illustrative, we do not restrict ourselves to it.  Instead, we consider a more generic class of wrinkled curves $\mathcal{A}_k^\tau$ which retains quantitative estimates on the amplification of curvature and higher order derivatives that we see in the examples with (\ref{eq:estimateEx}), but does not impose restrictions on the detailed appearance of the wrinkled curves: 
\begin{equation}
\label{eq:admissable}
\begin{aligned}
\mathcal{A}_k^\tau := \Big \{ \gamma_k  &\in C^3([0,1],\mathbb{R}^3)  \colon \gamma_k \text{ as in } (\ref{eq:gammak}),   \\
& \|\gamma_k - \bar{\lambda} x_2e_2\|_{L^{\infty}} \leq \tau k^{-1},  \quad \| \gamma_k''\|_{L^{\infty}} \leq \tau k, \quad  \|\gamma_k''' \|_{L^{\infty}} \leq \tau k^{2} \Big \}
\end{aligned}
\end{equation}
for some $\tau >0$ and $k \in \mathbb{Z}^+$ (a positive integer).   For a given $\tau >0$ and $k \in \mathbb{Z}^+$, this set encapsulates a rather generic class of wrinkled curves where we think of the index $k$ as denoting essentially the number of wrinkles (or $1/k$ as the approximate wavelength) and $\tau$ quantifying the localized features of the wrinkles -- with large $\tau$ allowing for more localized features.  Moreover for fixed $\tau >0$, any sequence of curves $\gamma_k \in \mathcal{A}_k^\tau$ is assured the convergence in (\ref{eq:gammakweak}) (up to a subsequence).  Hence, any sequence of midplane deformations $y_k$ in (\ref{eq:ykDef}) with wrinkled curves $\gamma_k \in \mathcal{A}_k^\tau$ has the desired limiting behavior in  (\ref{eq:sequenceImport}).  

%%%%%%%%%%%%%%%%%%%%%%%%%%%%%%%%%%%%%%%%%%%%%%%%
%%%%%%%%%%%%%%%%%%%%%%%%%%%%%%%%%%%%%%%%%%%%%%%%
\subsection{Incompressible extensions to the sheet}
Having characterized the tension wrinkling deformations of the midplane, we seek to extend them to the entire thin sheet.   We have to do so in such a manner that preserves incompressibility and yields low membrane energy.   Throughout this section we assume that $\tau>0$ is fixed and the integer $k$ and thickness are such that $hk < 1$.  

Given a midplane deformation $y_k$ of the form (\ref{eq:ykDef}) with $\gamma_k \in \mathcal{A}_k^\tau$, we define 
\begin{align}\label{eq:bkDef}
b_k := \frac{\partial_1 y_k \times \partial_2 y_k}{|\partial_1 y_k \times \partial_2 y_k|^2} = \delta(\tilde{\nabla} y_k)^{-1} \nu_{y_k} \quad \text{ on } \omega
\end{align} 
in light of the attainment result in Proposition \ref{W3Dqc2D} in Appendix \ref{sec:KoiterAppend1}.  We then extend the midplane deformation to $\Omega_h$ by setting 
\begin{align}\label{eq:GlobalIncomp1}
y_k^h(x) : =  y_k(\tilde{x}) + \xi_k^h(x) b_k(\tilde{x}), \quad x \in \Omega_h 
\end{align}
for some 
$\xi_k^h \in C^1(\overline{\Omega}_h,\mathbb{R}^3)$ satisfying 
\begin{equation} \label{eq:xi0}
\xi_k^h(\tilde{x},0) = 0.
\end{equation}
By an important result, Proposition \ref{IncompProp} stated in Appendix \ref{sec:KoiterAppend}, we can find a function $\xi_k^h$ which yields an incompressible extension if the sheet is sufficiently thin, i.e., 
\begin{align}\label{eq:GlobalIncomp2}
\det(\nabla y_k^h) = 1 \quad \text{ on } \Omega_h.
\end{align}
Further, the gradient of this extension satisfies the following approximation
\begin{equation}
\label{eq:ImportantForm}
 \begin{aligned}
\nabla y_k^h & =(\tilde{\nabla} y_k|b_k) + x_3(\tilde{\nabla} b_k|0) - x_3 \Tr(G_k) b_k \otimes e_3  \\
&\quad + \frac{3}{2} x_3^2 \Tr(G_k)^2b_k \otimes e_3 - \frac{1}{2} x_3^2 \Tr(G_k) )(\tilde{\nabla} b_k|0) \\ 
&\quad + O(k^2x_3^2) b_k \otimes e_2 + O(k^3 x_3^3) \quad \text{ on } \Omega_{h}, \\
 G_k &:= (\tilde{\nabla} y_k |b_k)^{-1} (\tilde{\nabla} b_k|0) \quad \text{ on } \omega. \\
\end{aligned}
\end{equation}
The derivation of this deformation makes use of techniques due to Conti and Dolzmann   \cite{cd_06_chap,cd_cov_09} (also see Plucinsky {\it {\it et al.}\ }\cite{plb_arxiv_16}).  Finally, the deformation satisfies other properties, which are listed in 
Propositions \ref{EigenProp} and \ref{PropIdentities} in Appendix \ref{sec:KoiterAppend}.

Note that the form (\ref{eq:GlobalIncomp1}) can be expanded in $x_3$ as $y_k^h = y_k + x_3 b_k - \frac{x_3^2}{2} \Tr(G_k) b_k + \ldots$, and so this deformation is consistent with the leading order behavior in (\ref{eq:ext}). 
\subsection{The energy of tension wrinkling}\label{ssec:TensEnergy}
We now compute the energy of the deformations (\ref{eq:GlobalIncomp1}).  Throughout, we make the physically reasonable restriction to deformations for which the radius of curvature is large compared to the thickness, which is tantamount to assuming $kh \ll 1$.  In this setting, we note from Proposition \ref{EigenProp} in Appendix \ref{sec:KoiterAppend} that
\begin{align}\label{eq:IdentLambdaM}
&(\lambda_M(\tilde{\nabla} y_k|b_k), \lambda_M(\cof (\tilde{\nabla} y_k |b_k)) )= (\bar{\lambda}_M, \bar{\lambda}_M^{1/2}) \in S \quad \text{ on } \omega,\\
&\lambda_M(\nabla y_k^h) =  \bar{\lambda}_M  \quad \text{ and } \quad (\lambda_M(\nabla y_k^h), \lambda_M(\cof \nabla y_k^h)) \in S \quad \text{ on } \Omega_h \label{eq:IdentLambdaM2}.
\end{align}
Further, for any $F \in \mathbb{R}^{3\times3}$ with $(\lambda_M(F), \lambda_M(\cof F)) \in S$,
\begin{align}
W_{3D}^{qc} (F) = W_{3D} (F) = \frac{\mu}{2} \left( r^{1/3} 
\left( |F|^2 - \left(\frac{r-1}{r} \right) (\lambda_M(F))^2 \right)  -3 \right).
\end{align}
So in setting
\begin{align}\label{eq:Akh}
A_k^h := \nabla y_k^h - (\tilde{\nabla} y_k |b_k) \quad \text{ on } \Omega_h,
\end{align}
we obtain 
\begin{equation}
\label{eq:energyCalc1}
\begin{aligned}
\mathcal{E}_{3D}^h(y_k^h) &= \frac{\mu}{2} \int_{\Omega_h} \left(r^{1/3}\left( |\nabla y_k^h|^2 - \left(\frac{r-1}{r} \right) (\lambda_M(\nabla y_k^h) )^2 \right) - 3\right) dx  \\
&= \int_{\Omega_h} \left( W_{3D}^{qc}((\tilde{\nabla} y_k|b_k)) + \frac{\mu  r^{1/3}}{2} \left( 2 \Tr((\tilde{\nabla} y_k |b_k)^T A_k^h) + |A_k^h|^2 \right) \right) dx \\
&= h \int_\omega W_{ps} (\tilde{\nabla} y_k) d \tilde x + 
\frac{\mu  r^{1/3}}{2} \int_{\Omega_h} \left( 2 \Tr((\tilde{\nabla} y_k |b_k)^T A_k^h) + |A_k^h|^2 \right) dx 
\end{aligned}
\end{equation}
using (\ref{eq:IdentLambdaM}) and Proposition \ref{W3Dqc2D} combined with the definition of $b_k$ in (\ref{eq:bkDef}).  

We compare (\ref{eq:ImportantForm}) and (\ref{eq:Akh}) to expand $A_k^h$.  Substituting this expansion into (\ref{eq:energyCalc1}), we see first that 
\begin{equation}
\label{eq:energyCalc3}
\begin{aligned}
\mu r^{1/3} \int_{\Omega_h} &\Tr((\tilde{\nabla} y_k|b_k)^T A_k^h) dx =\\
& \frac{r^{1/3}\mu h^3}{24} \int_{\omega} \left( 3 \Tr(G_k)^2 |b_k|^2 - \Tr(G_k) \Tr((\tilde{\nabla} y_k|b_k)^T (\tilde{\nabla} b_k|0)) \right) d\tilde{x} + O(k^3h^4)
\end{aligned}
\end{equation}
using the fact that the terms linear in $x_3$ vanish upon integration through the thickness and also using (\ref{eq:PropIdent1}) in Proposition \ref{PropIdentities} in Appendix \ref{sec:KoiterAppend}.  By a similar argument and with the additional identity (\ref{eq:PropIdent2}) in this proposition, we arrive at 
\begin{align}\label{eq:energyCalc4}
\frac{\mu r^{1/3} }{2} \int_{\Omega_h} |A_k^h|^2 dx = \frac{\mu r^{1/3} h^3}{24} \int_{\omega} \left(|\tilde{\nabla} b_k|^2 + \Tr(G_k)^2 |b_k|^2 \right)d\tilde{x} + O(k^3h^4).
\end{align}
Making use of the remaining identities in this proposition, we find that the two energies (\ref{eq:energyCalc3}) and (\ref{eq:energyCalc4}) combine to form a term penalizing the second fundamental form of the midplane deformation $y_k$,
\begin{align}\label{eq:energyCalc5}
\frac{\mu r^{1/3}}{2} \int_{\Omega_h} \left(2\Tr((\tilde{\nabla} y_k|b_k)^T A_k^h)  + |A_k|^2 \right)dx = \frac{\mu r^{1/3}h^3}{6} \int_{\omega} |\II_{y_k}|^2 d\tilde{x} + O(k^3h^4).
\end{align} 

Finally in combining (\ref{eq:energyCalc1}) and (\ref{eq:energyCalc5}), we obtain 
\begin{align}
\mathcal{E}_{3D}^h(y_k^h) = \mathcal{E}_{K}^h(y_k) + O(k^3 h^4).
\end{align}
Observe that the stretching term in $\mathcal{E}_{K}^h$ is $O(h)$ while the bending part is $O(k^2 h^3)$.  Thus, with the radius of curvature large compared to the thickness, $kh \ll 1$ and therefore the term $O(k^3 h^4)$ is negligible compared to the rest.

%%%%%%%%%%%%%%%%%%%%%%%%%%%%%%%%%%%%%%%%%%%%%%%%
%%%%%%%%%%%%%%%%%%%%%%%%%%%%%%%%%%%%%%%%%%%%%%%%
\section{Numerical implementation}\label{sec:Numerical}

\subsection{The Koiter theory}\label{ssec:NumImpKoit}

We implement the commercial software package ABAQUS via a user material model (UMAT) for S4 shell elements.  This requires Cauchy stress to be specified as a function of the $2 \times 2$ surface deformation gradient $\widehat{F} \in \mathbb{R}^{2\times2}$ as well as the consistent tangent modulus derived from this Cauchy stress (see, for instance, the Abaqus User Subroutines Reference Manual, the section on UMAT, and the subsection on large volume changes with geometric nonlinearity \cite{abq_610}).   For this formulation, we take the energy as in (\ref{eq:WpsDef}) restricted to $2 \times 2$ matrices $\tilde{F} = \widehat{F}$ where $\tilde{F}_{31} = \tilde{F}_{32} = 0$.  This corresponds to fixing a frame.  Also, $\delta(\tilde{F}) = \det \widehat{F}$, $\lambda_M(\tilde{F}) = \lambda_M(\widehat{F})$ and following the formalism of Cesana {\it et al.} \cite{cpk_arma_15}, the Cauchy stress is given by
\begin{align}\label{eq:CauchyK}
\sigma^{K}_r(\widehat{F}) :=  (W_{ps}),_{\widehat{F}} \widehat{F}^T 
\end{align}
where the subscript $r$ is used to emphasize the dependence of this stress on the anisotropy parameter.

Now, the membrane stress $\sigma_r^K$ in (\ref{eq:CauchyK}) inherits soft elasticity, a feature distinct to nematic elastomers as compared to purely elastic materials.  This poses some numerical challenges.
At small strains, much of the deformation of these elastomers can be accommodated by either soft or very lightly stressed microstructure (regions $\mathcal{L}_m$ and $\mathcal{M}$ on the energy landscape in Figure \ref{fig:KEnergy}).  For these regions, the tangent modulus derived from $\sigma_{r}^K$ in  (\ref{eq:CauchyK}) is positive semi-definite and not strictly positive definite: a nematic elastomer has zero stiffness in region $\mathcal{L}_m$ and it has no stiffness against shear in region $\mathcal{M}$.  These features lead to difficulties in numerical convergence at small strains.   To combat this, we introduce a small energy regularizer whose Cauchy stress has the form
\begin{align}\label{eq:Wepsilon}
\sigma_{\epsilon}^{reg}(\widehat{F}) := \mu \epsilon \left( \widehat{F} \widehat{F}^T - I_{2\times2} \right)
\end{align}
for $0 \leq \epsilon \ll 1$.  The consistent tangent modulus derived from $\sigma_{\epsilon}^{reg}$ is strictly positive definite for all surfaces deformation gradients $\widehat{F} \in \mathbb{R}^{2\times2}$ with $\det \widehat{F} > 0$.  Thus, introducing the regularizer $\epsilon \ll 1$ stabilizes the numerical calculation at small strains.   For our calculations we use $\epsilon = 0.05$.  We have studied the role of $\epsilon$ on our simulations and only present conclusions that are independent of this regularization.

The formulas derived in Sections \ref{sec:TheoryBack} and \ref{sec:k} take the isotropic state to be the reference.  However, notice that identity is deep in the interior of the soft region and $\sigma_{r}^K = 0$ in a neighborhood of identity for $r > 1$ (identity corresponds to the intersection of the red and dark blue curve in on the energy landscape for $W_{ps}$ in Figure \ref{fig:KEnergy}(a)).  This is a trivial soft regime, which is not of interest here.\footnote{We refer to Conti {\it {\it et al.}} \cite{cdd_02_pre,cdd_02_jmps}  for extensive investigations of soft elasticity in a similar framework.}  We are interested in the interplay of microstructure and tension wrinkling, an interplay which is broached in this setting when a genuine tensile response to stretch is induced in the membrane.   Therefore, we take the end of this trivial soft state, one with  $\lambda_M = \delta = r^{1/6}$ as the reference state by setting 
\begin{align}\label{eq:CauchyStress}
\sigma^{sim}_{r,\epsilon} (\widehat{F}) := \sigma^{K}_r(\widehat{F} \widehat{U}_r) + \sigma^{reg}_{\epsilon}(\widehat{F}), \quad \widehat{U}_{r} := \left(\begin{array}{cc} r^{1/6} & 0 \\ 0 & 1\end{array}\right).
\end{align} 

Finally, the Abaqus S4 shell element slightly modifies the bending term from that derived in our Koiter theory.
Specifically, the shell element is based on a kinematic ansatz of the shell deformation through the thickness (see the Abaqus Theory Manual, 3.6.5 Finite-strain shell element formulation \cite{abq_612}).   This formulation, in effect, yields a bending energy which depends on a linearization of the second fundamental form.

To study wrinkling, we use the two part procedure detailed by Wong and Pellegrino \cite{wp_06_jmms} and Ling Zheng \cite{z_08_wrinkle} to capture the wrinkling geometry in Abaqus.  This involves first a pre-buckling eigenvalue analysis, and then a post buckling analysis.  For the pre-buckling eigenvalue analysis, we introduce a small initial prestress in the form of an edge displacement to the membrane.  Then we perform an eigenvalue buckling analysis on this lightly stressed membrane to determine the likely buckling modes.   After computing the buckling mode shapes, we introduce a linear combination of one or more selected eigenmodes as a geometric imperfection in the membrane at the start of the post buckling analysis.   We then apply an initial prestress to the membrane (as in the pre-buckling analysis) in order to provide an initial out-of-plane stiffness to the membrane.  Finally, we use the static stabilization procedure in Abaqus to compute the wrinkled shape of the membrane under the clamped stretching deformation.  Our detailed implementation of this procedure follows the input file example (Ling Zheng \cite{z_08_wrinkle} in Appendix B) modified appropriately to incorporate the user material model (\ref{eq:CauchyStress}) for nematic elastomer sheets.

\subsection{Membrane (tension field) theory}

We compare the results of simulations under the implementation of the Koiter theory described above with analogous simulations of the membrane theory for nematic elastomers (described in Section \ref{ssec:EffectiveTheory}).  We implement this once again into 
%. This theory has a strain energy density $W_{2D}^{qc}$ in (\ref{eq:W2Dqc}) (also Figure \ref{fig:2DEnergies}(b) and (c)), and it captures instabilities such as wrinkling, microstructure and crumpling through effective planar deformation and stress.  In particular, this theory is a tension field theory void of any compressive stresses, and so the region of compression $\mathcal{C}$ inherent to the membrane stress term $(W_{3D}^{qc})_{2D}$ in the Koiter theory is replaced by a region of zero stress $\mathcal{L}$ corresponding to effective planar deformation to capture crumpling and a region of uni-axial stress $\mathcal{W}$ corresponding to effective planar deformation to capture wrinkling.  For the clamped stretched geometry under this theory, regions of effective deformation in $\mathcal{W}$, which emerge in equilibrium under stretch, offer a qualitative picture for the extent of wrinkling in these membranes.  By comparison, wrinkling emerges in the Koiter theory through a balance of bending energy and membrane energy induced by compressive stresses.   This contrast offers an interesting comparison of the wrinkling phenomenon in these sheets.  Alternatively, the two theories are, effectively, identical in their means of capturing microstructure for these simulations, thus facilitating another point of comparison.
%We implement the effective membrane theory into 
ABAQUS via UMAT for plane stress CPS4 element, where we take the Cauchy stress in the implementation to be akin to (\ref{eq:CauchyStress}), with
\begin{align}
\sigma_{r}^{sim}(\widehat{F}) := \sigma_{r}^{mem}(\widehat{F} \widehat{U}_r), \quad \sigma_{r}^{mem}(\widehat{F}) := (W_{2D}^{qc}),_{\widehat{F}} \widehat{F}^T.  
\end{align} 
Notice here that we do not employ the use of a regularizer.  For simulating the effective membrane theory, there is no wrinkling bifurcation to compute, as wrinkling is treated through effective deformation.   Consequently, we do not need to explore the stretch monotonically from small initial stretch.  Instead, we explore a range of stretch which excludes small stretch associated to soft elasticity.  This can be done without regularization (i.e., by first deforming the specimen homogeneously from the unstressed reference state to a state of stretch far away from soft elastic behavior with the top and bottom of the sheet fixed in the $e_2$ direction, and then relaxing the boundary conditions on top and bottom to a sheet which is traction free away from the clamped ends).

\section{Results: Microstructure induced suppression of wrinkling} \label{sec:results}

We now study the clamped stretch experiments of Kundler and Finkelmann \cite{kf_mrc_95}.    We consider rectangular sheets with length $254$ mm, width $101.6$ mm and thickness $h = 0.1$ mm  that is clamped on both ends and subject to a stretch along its length in the $e_1$ direction.  We choose these dimensions since purely elastic sheets of this dimension readily wrinkle when stretched as demonstrated experimentally and numerically by Zheng \cite{z_08_wrinkle} (also confirmed numerically by Nayyar {\it {\it et al.}\ }\cite{nrh_11_ijss} by Taylor {\it {\it et al.}\ }\cite{tbs_14_jmps}).  We also fix the shear modulus to be $\mu = 6\times 10^5$ Pa in accordance to experimental observations \cite{wt_lceboox_03}, though this choice is irrelevant as all terms in the energy and stress scale with $\mu$.  We take $r$ to be in the range $1$ through $1.45$ to explore the entire elastic to nematic range.  

We often use the nominal strain 
\begin{align}\label{eq:epsEng}
\varepsilon_{eng} := \frac{L_{final} - L_{initial}}{L_{initial}}
\end{align}
to display the results.   We also use a microstructure indicator parameter
\begin{align}\label{eq:MicroParam}
\frac{ r^{1/2} \delta}{\lambda_M^2}.
\end{align}
Recall that the curve $\lambda_M^2 = r^{1/2} \delta$ forms the boundary between the region ${\mathcal M}$ with microstructure and region ${\mathcal S}$ without.  So, a microstructure parameter (\ref{eq:MicroParam}) larger than $1$ indicates the presence of fine-scale microstructure.

%%%%%%%%%%%%%%%%%%%%%%%%%%%%%%%%%%%%%%%%%%%%%%%%
%%%%%%%%%%%%%%%%%%%%%%%%%%%%%%%%%%%%%%%%%%%%%%%%
\subsection{Simulations with Koiter theory}

\begin{figure}[t]
\centering
\begin{subfigure}{\linewidth}
\centering
\includegraphics[width =4.5in]{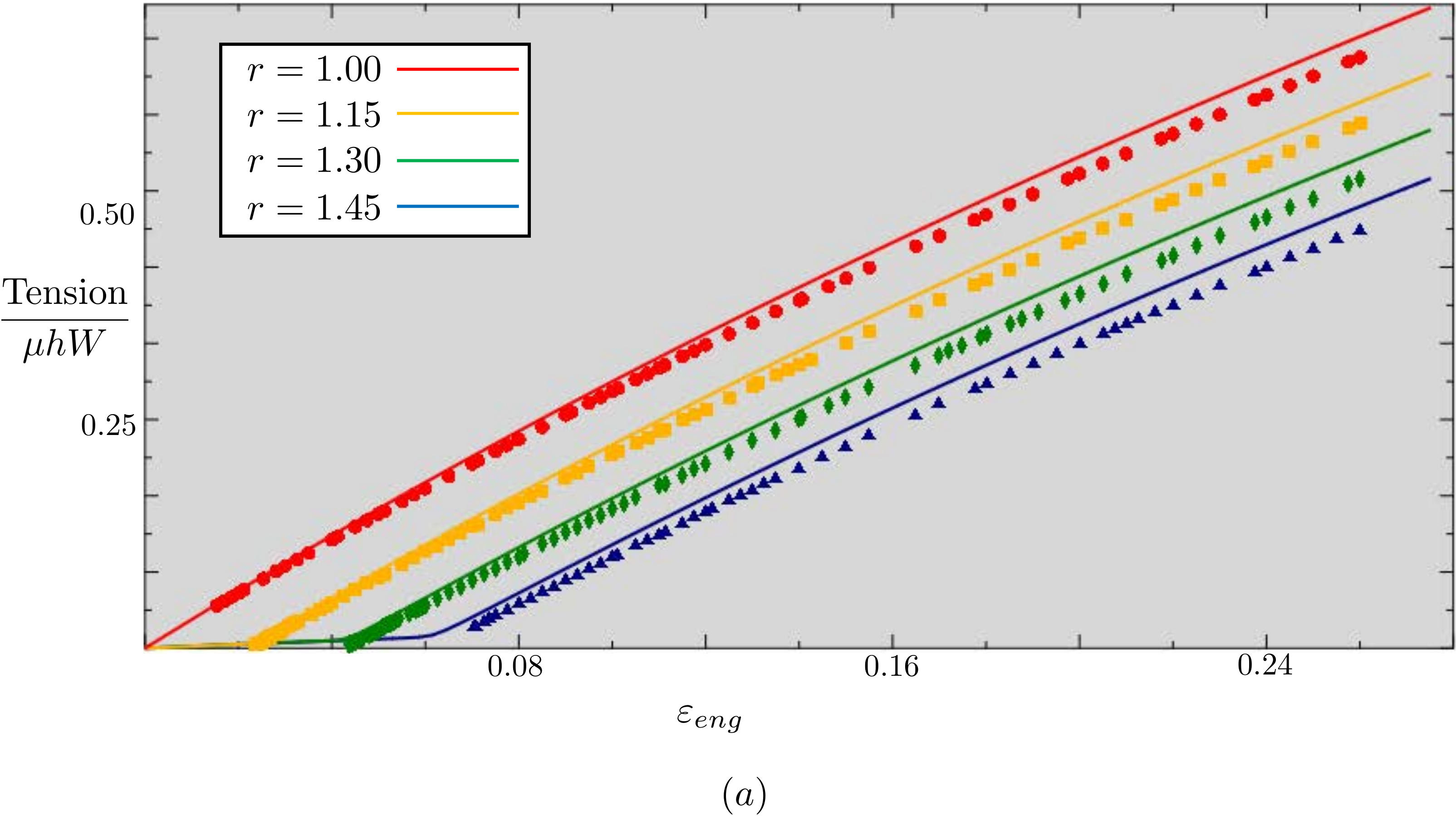}
\end{subfigure}
\begin{subfigure}{\linewidth}
\centering
\includegraphics[width = 4.5in]{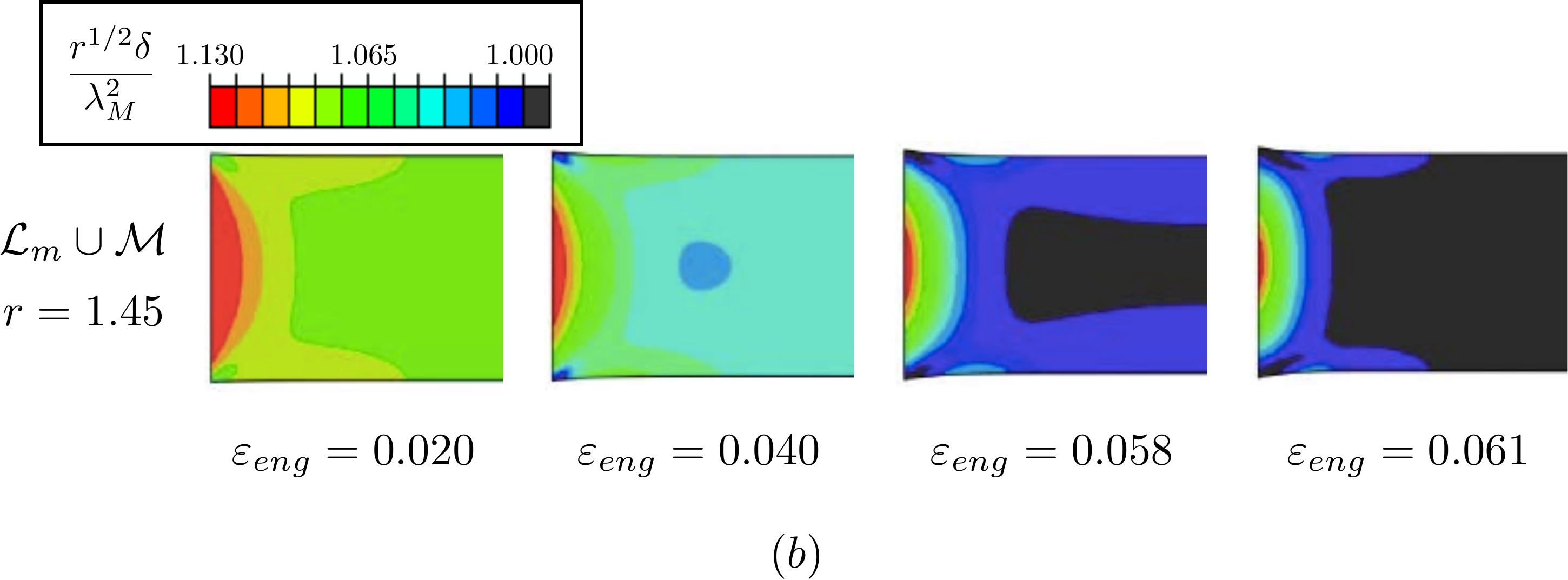}
\end{subfigure}
\caption{(a) Normalized tensile response to stretch (with $W = 101.6$ mm the undeformed width) for nematic elastomer membranes of varying anisotropy: Koiter theory simulations (lines), Effective membrane theory simulations (points).  (b) Soft elastic response for $r= 1.45$ nematic membrane, and transition to stressed elasticity.  Microstructure in $\mathcal{L}_m \cup \mathcal{M}$ indicated by colored regions ($r^{1/2} \delta/\lambda_M^2 > 1$).  Note that only the left half of the sheet is shown as the distribution is symmetric}
\label{fig:Soft}
\end{figure}

The nominal tensile force of the sheet is plotted in Figure \ref{fig:Soft}(a) as a function of the nominal strain.  The purely elastic sheet ($r=1$) is, as expected, immediately tensioned due to the stretch, whereas the nematic sheets ($r > 1$) are soft and nearly stress free during the initial stages of stretch.  We call the extent of the soft strain $\varepsilon_{eng}^{soft}$.   This depends on the nematic anisotropy $r$.   Recall that we start at $\widehat{U}_r$ or $(\lambda_M, \delta) = (r^{1/6}, r^{1/6})$, the far left point on the boundary between ${\mathcal M}$ and ${\mathcal L}_m$.  On the initial application of strain, we expect much of the membrane to traverse this boundary with no stress to the right until it reaches the point $(\lambda_M, \delta) = (r^{1/3}, r^{1/6})$.  So we expect to see soft behavior until $\varepsilon_{eng}^{soft} \approx r^{1/6} - 1$.  This is consistent with the simulations (e.g., the formula gives $\varepsilon_{eng}^{soft} = 0.063$ in agreement with $0.061$ in the simulations for $r=1.45$).  We note that the soft elastic strain in the Kundler and Finkelmann experiments is $r^{1/2} - 1$ since they start from a state where the director is uniformly vertical instead of $\widehat{U}_r$.

The formation of the microstructure is indicated in Figure \ref{fig:Soft}(b) for the case $r=1.45$ by plotting the distribution of the microstructure parameter (\ref{eq:MicroParam})  (note that only the left half of the sheet is shown as the distribution is symmetric).  Since we start from the state at $\widehat{U}_r$, the initial sheet has uniform microstructure.  As we stretch, the microstructure evolves differently close to the grips compared to away from them: it is gradually driven out in most of the sheet but persists at the grips.  Since the imposed deformation is accommodated by the rearrangement of the microstructure, the response is soft.  Eventually, all the microstructure is driven out and director is uniformly horizontal except close to the grips, signaling an end of the soft behavior.  Subsequent stretching leads to proportional increase in load.  All of this is consistent with the observations of Kundler and Finkelmann \cite{kf_mrc_95} and the simulations of Conti {\it {\it et al.}\ }\cite{cdd_02_jmps}.

\begin{figure}
\centering
\begin{subfigure}{\linewidth}
\centering
\includegraphics[width =6in]{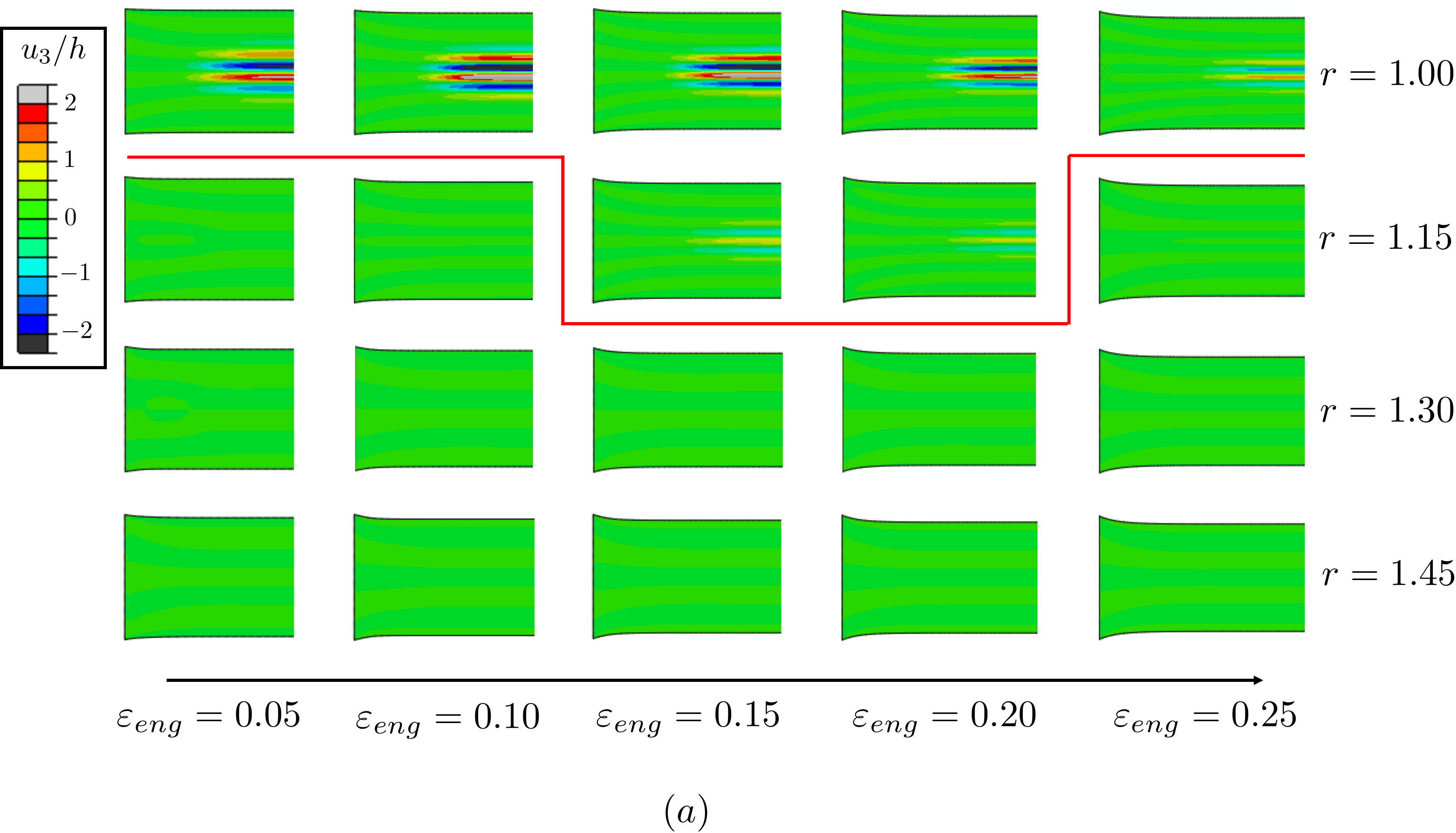}
\end{subfigure}
\begin{subfigure}{\linewidth}
\centering
\includegraphics[width = 6.0in]{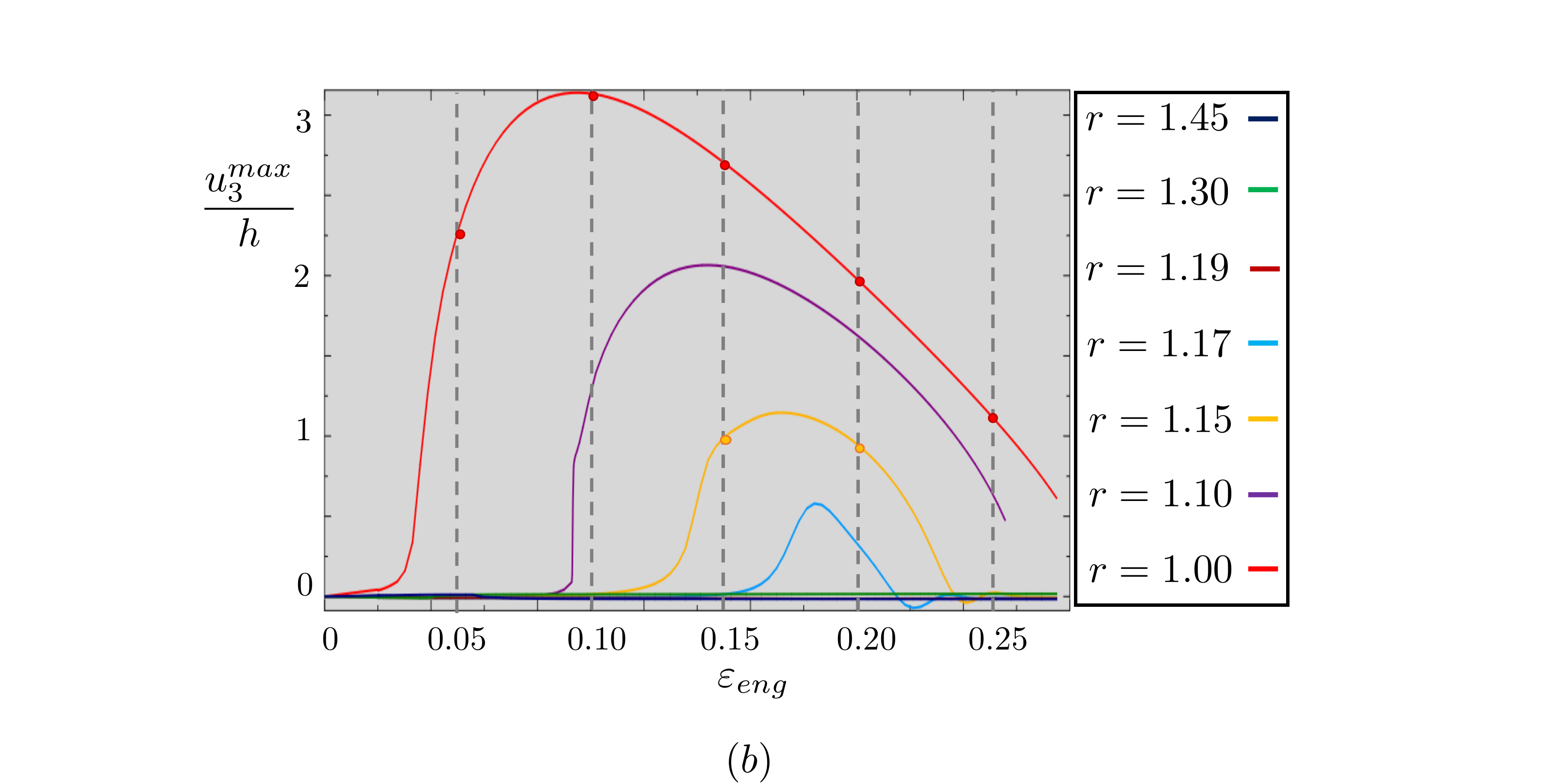}
\end{subfigure}
\caption{The evolution of wrinkles with applied stretch for a purely elastic material and nematic elastomers with increasing degrees of order.  (a) Snapshots of the out-of-plane displacement.  (b) The amplitude of wrinkles (i.e., maximum out of plane displacement) as a function of stretch for various values of $r$.  We see that wrinkles appear early, have large amplitude and disappear late for the usual elastic material ($r=1$), but are suppressed with the introduction of nematic order.}
\label{fig:Wrinkle}
\end{figure}

We show the evolution of wrinkles by plotting the out-of-plane displacement in Figure \ref{fig:Wrinkle} for a variety of $r$ and $\varepsilon_{eng}$.  Notice from the top row of Figure \ref{fig:Wrinkle}(a) that for the purely elastic sheet ($r=1$), wrinkles appear almost immediately upon stretch, grow with further stretch, reach a maximum amplitude $u_3 /h \approx 3$ and eventually diminish.  This, in fact, reproduces the results of Zheng \cite{z_08_wrinkle} and Nayyar {\it {\it et al.}\ }\cite{nrh_11_ijss}.  For a nematic sheet with small nematic order ($r=1.15$) shown on the second row, we see that wrinkles do not appear until a larger value of stretch, are much smaller in amplitude and disappear faster.  For higher values of $r$ shown in the third and fourth rows, wrinkles do not appear.  All of this is explored further in Figure \ref{fig:Wrinkle}(b), which shows the amplitude of the wrinkles as a function of stretch for various $r$.  We see that wrinkles appear early, have large amplitude and disappear late for the purely elastic sheet ($r=1$).  We also see that wrinkling is suppressed by the introduction of nematic order.  In fact, increasing $r$ leads to delayed onset, smaller amplitude and earlier disappearance of wrinkles.  Moreover, wrinkling is fully suppressed for values of $r$ greater than $1.2$.

\begin{figure}
\centering
\begin{subfigure}{\linewidth}
\centering
\includegraphics[width =6in]{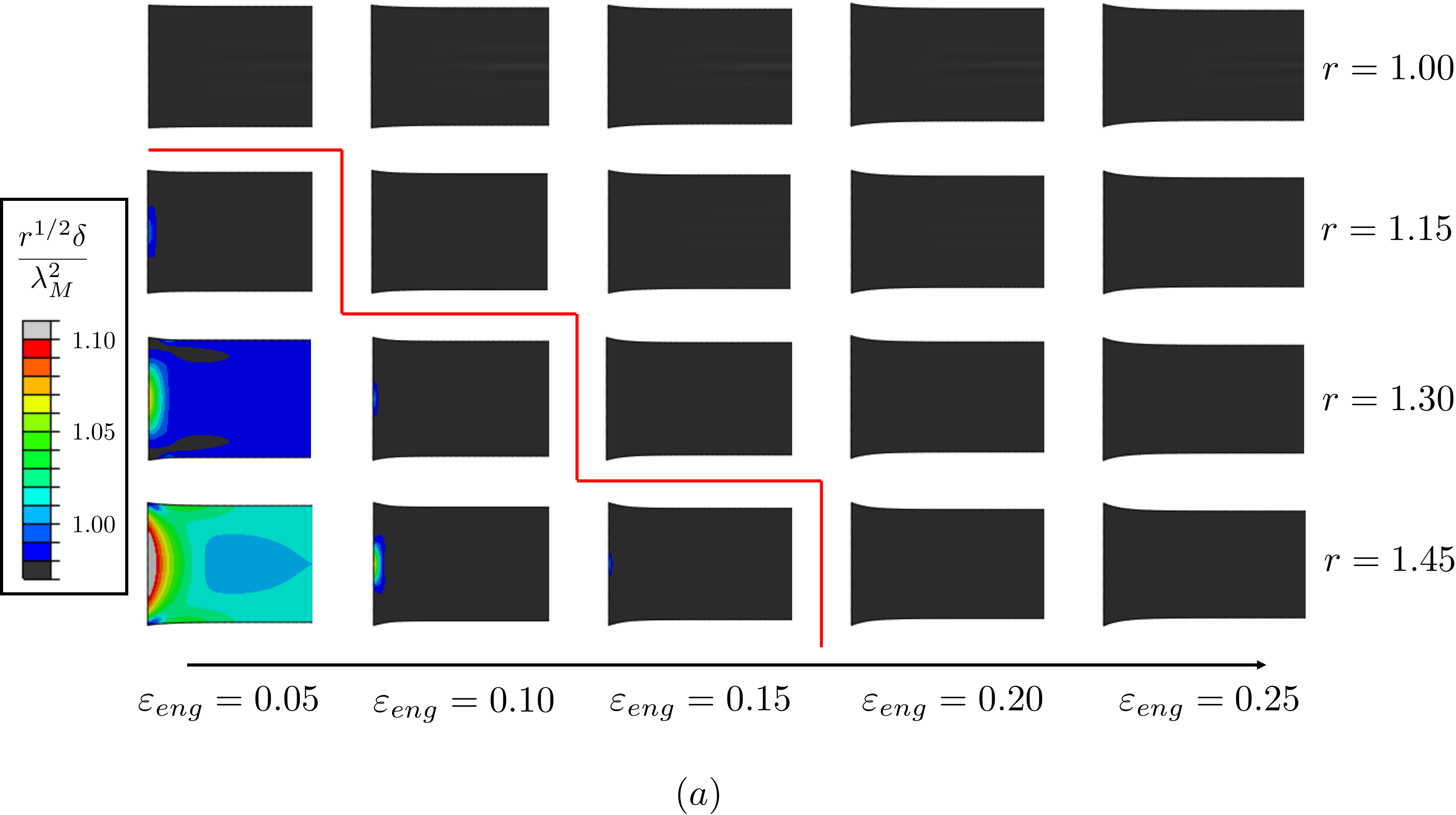}
\end{subfigure}
\begin{subfigure}{\linewidth}
\centering
\includegraphics[width = 6.0in]{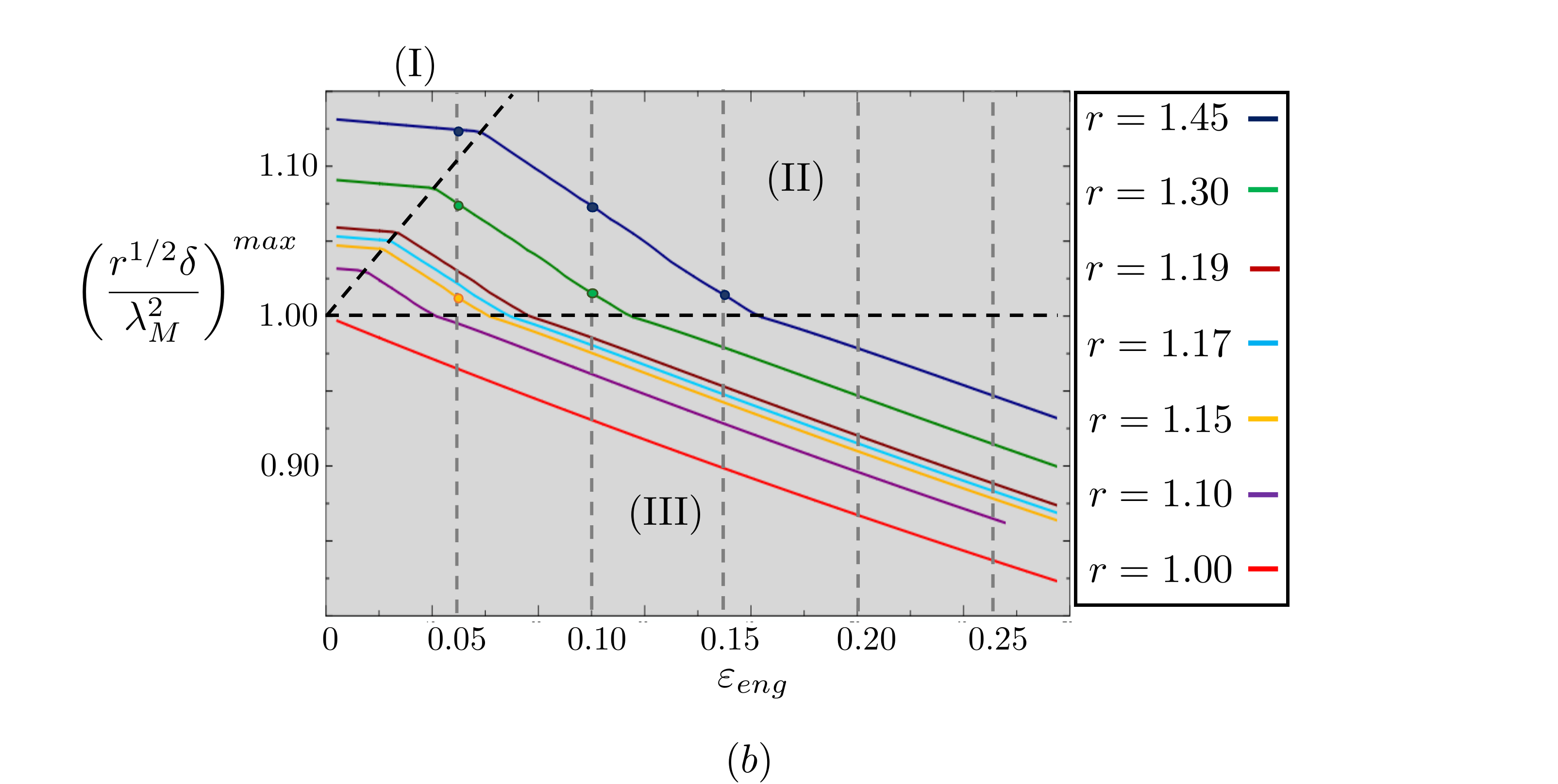}
\end{subfigure}
\caption{The evolution of microstructure with applied stretch.  (a) Snapshots of the microstructure parameter $r^{1/2} \delta/\lambda_M^2$ (values greater than 1 indicate microstructure).  (b) The maximum value of the microstructure parameter as a function of stretch for various values of $r$. This shows three distinct regions: (I) soft; (II) stressed but with microstructure; (III) without microstructure. }
\label{fig:Micro}
\end{figure}

Figure \ref{fig:Micro} shows the corresponding evolution of the microstructure.  We see from Figure \ref{fig:Micro}(a) that microstructure disappears in most of the sheet at small to modest stretch but persists for larger stretch near the grips.  We see from Figure \ref{fig:Micro}(b) that the maximum microstructure parameter $r^{1/2} \delta/\lambda_M^2$, i.e., the value at the grips, delineates three distinct regions describing microstructure in the sheet for the full range of stretch: (I) captures soft response to stretch of the sheet, (II) captures the stretch for which the sheet is stressed but with microstructure at the grips and (III) captures the stretch for which there is no longer any microstructure present in the sheet.  Thus, larger $r$ implies both a more prolonged soft elastic response and a more prolonged response for which microstructure persists at the clamps. 

\begin{figure}
\centering
\includegraphics[width =6in]{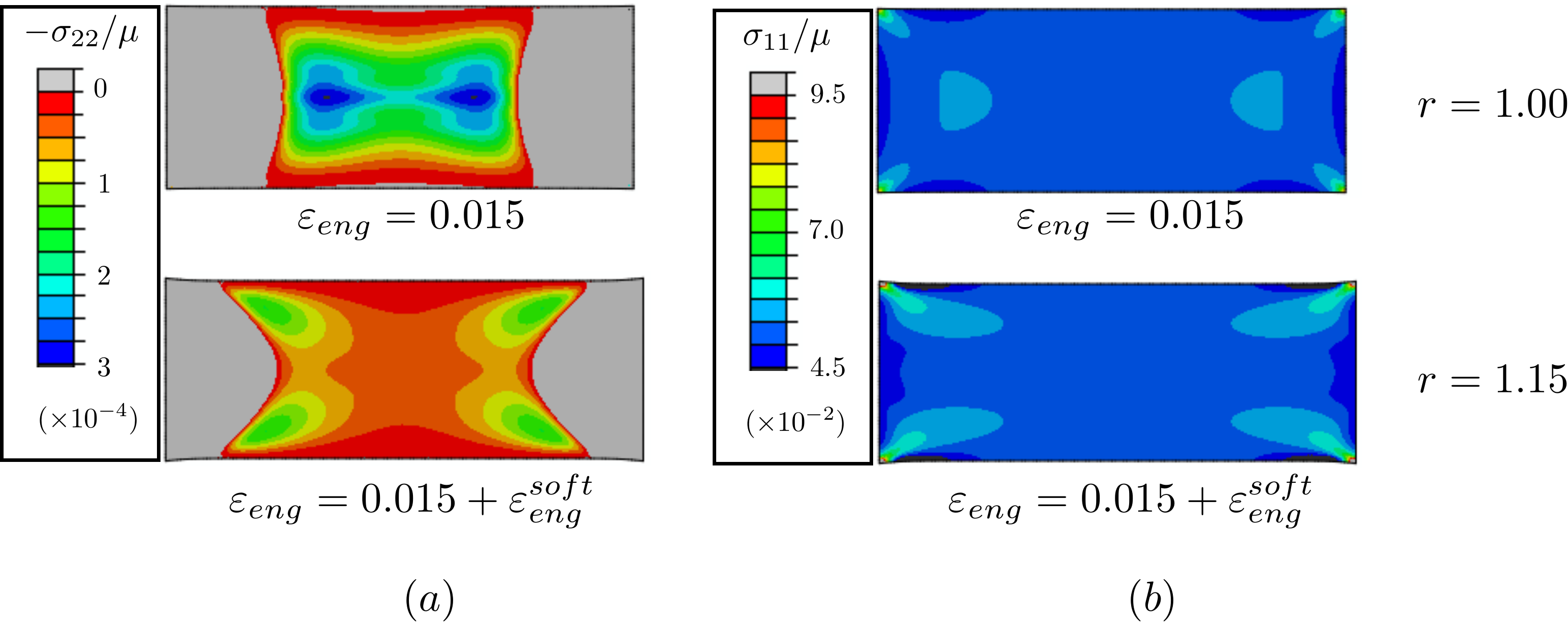}
\caption{The state of stress in a purely elastic sheet ($r=1$, top row) compared with that in a nematic sheet ($r=1.15$, bottom row) at comparable values of nominal load but before the onset of wrinkles in the elastic sheet.  Note that the purely elastic sheet shows significant transverse compression in the middle, but the nematic sheet does not.}
\label{fig:CompTen}
\end{figure}

We turn now to understanding the mechanism by which the presence of nematic microstructure suppresses wrinkling.  Figure \ref{fig:CompTen} compares the state of stress in a purely elastic sheet ($r=1$) with that of a nematic sheet ($r=1.15$) at comparable values of  nominal load but before the onset of wrinkles in the elastic sheet (i.e., at a value of $\varepsilon_{eng}$ of $0.015$ for the elastic sheet and $0.015$ beyond the soft strain for the nematic sheet).  Notice (top of Figure \ref{fig:CompTen}(a)) that the purely elastic sheet develops transverse compression.  This compressive stress leads to the buckling to a wrinkled state beginning at a stretch  $\varepsilon_{eng} \approx 0.025$.  However (bottom of Figure \ref{fig:CompTen}(a)), we notice that there is no transverse compression in the nematic sheet at a comparable value of nominal load.  There is some transverse compression closer to the edges, but it is much smaller in magnitude.  We see from Figure \ref{fig:CompTen}(b) that the distribution of longitudinal stress is also different.  All of this is a result of the persistence of the microstructure near the grips.  In other words, the ability of the nematic material to form microstructure not only gives rise to soft behavior in the bulk, but also qualitatively changes the distribution of stresses near the grips and that in turn suppresses the wrinkling.

Mechanistically, note that as the sheet is stretched, it seeks to compress laterally.  It is free to do so away from the grips, but is prevented from doing so near the clamped grips.  This leads to a shear deformation near the corners extending diagonally into the sheet.  This, in turn, leads to the compressive stress in the elastic sheet.  In contrast, the ability of the nematic sheet to form microstructure enables it to accommodate the shear strain through equi-biaxial tension at the clamps (recall that nematic sheets can accommodate shear strain without shear stress in the region or microstructure ${\mathcal M}$). 

\begin{figure}
\centering
\includegraphics[width =6in]{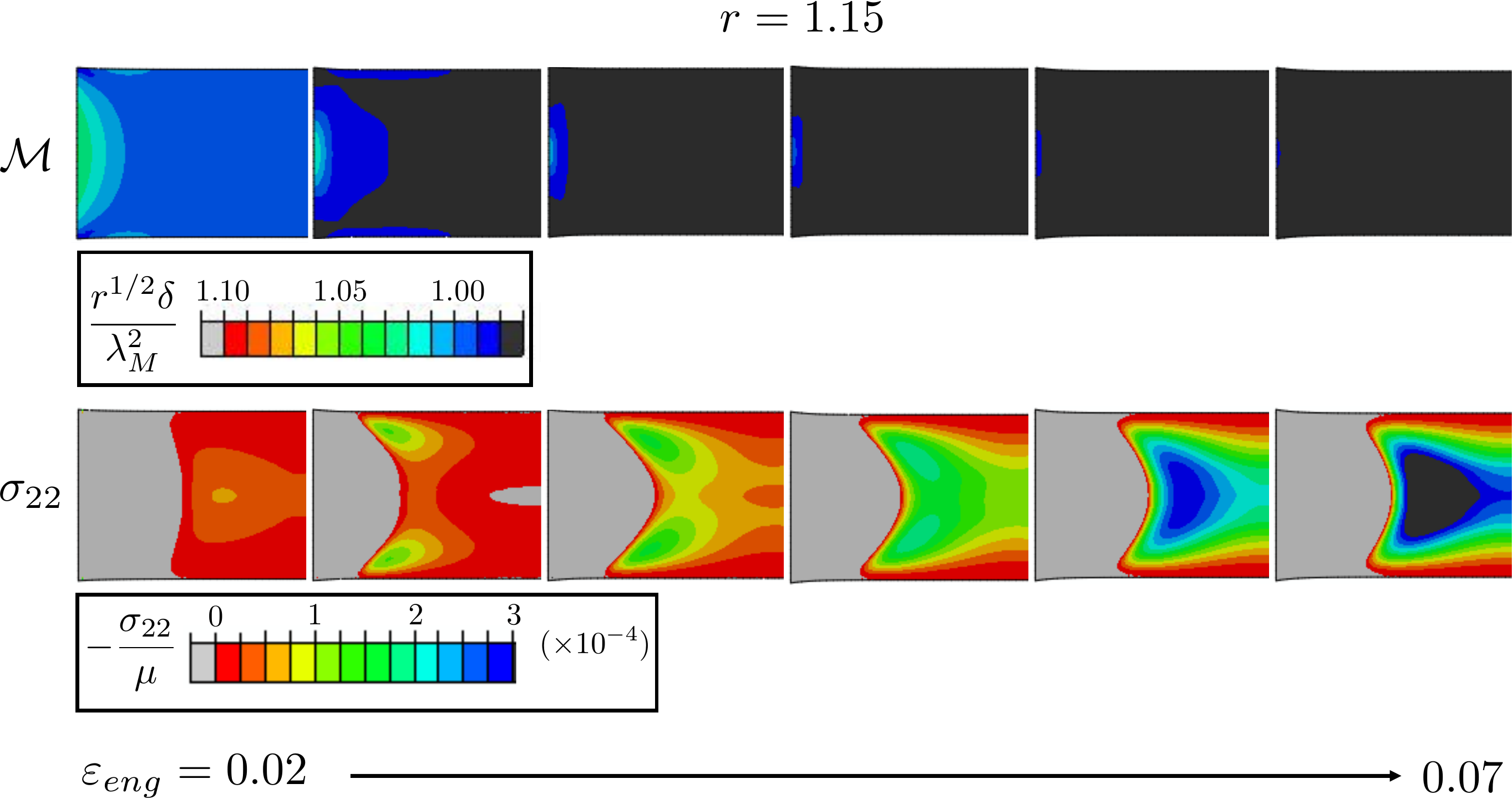}
\caption{The evolution of the microstructure and the lateral compressive stress with increasing stretch in a sheet with $r=1.15$.  Note that there is little to no compression in the middle as long as microstructure persists near the grips.}
\label{fig:CompTransition}
\end{figure}

\begin{figure}
\centering
\includegraphics[width =6in]{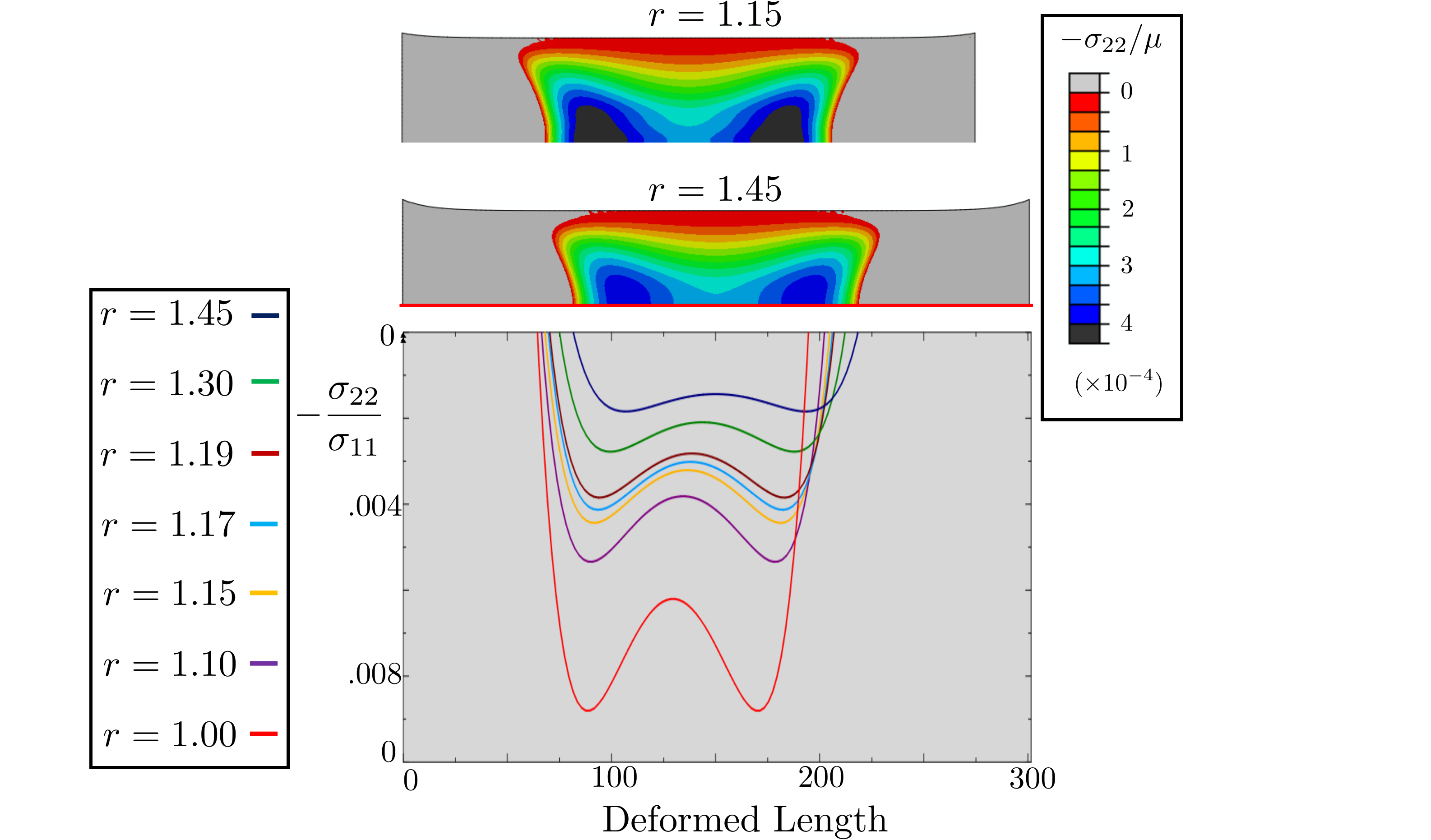}
\caption{The distribution of the lateral compressive stress (top) and the ratio of the lateral compressive stress to longitudinal tensile stress along the middle line in sheets with various $r$ at the stretch where the microstructure just disappears at the clamps.}
\label{fig:NoCompExplain}
\end{figure}

Figures \ref{fig:CompTransition} and \ref{fig:NoCompExplain} elaborate on this. Figure \ref{fig:CompTransition} shows how the microstructure parameter and transverse compression evolve in the sheet with increasing stretch, starting at the end of soft behavior.  We observe that the microstructure persists near the grips but reduces with increasing stretch until it is fully driven out.  Additionally, we see small regions of transverse compressive stress form near the free edges, and these region gradually moves inward with increasing magnitude as microstructure is driven out.  Finally, when all microstructure is driven out, we have transverse compressive stress in the middle similar to that of the purely elastic sheet (top of Figure \ref{fig:CompTen}(a)).  This compression increases with further stretch and eventually leads to wrinkling (Figure \ref{fig:Wrinkle}).  However, as shown in Figure \ref{fig:NoCompExplain}, the ratio of the transverse compression to the longitudinal tension at a value of stretch where all microstructure is driven out decreases with increasing $r$.   As we know from the study of elastic sheets, high relative longitudinal tension suppresses wrinkles\footnote{This is also evident in the role of aspect ratio as was shown by Ling Zheng \cite{z_08_wrinkle} and Nayyer {\it et al.} \cite{nrh_11_ijss}.} -- note from Figure \ref{fig:Wrinkle}(a) that wrinkles disappear in the elastic sheet at high stretch.   Thus, the microstructure near the grips delays the formation of a central region of transverse compressive stress and the onset of wrinkles.  However, large relative longitudinal tension also suppresses wrinkles; so if the delay is sufficient as is the case for larger $r$, then wrinkling is fully suppressed.

\subsection{Simulations with effective membrane theory and comparison}

\begin{figure}
\centering
\includegraphics[width =6.0in]{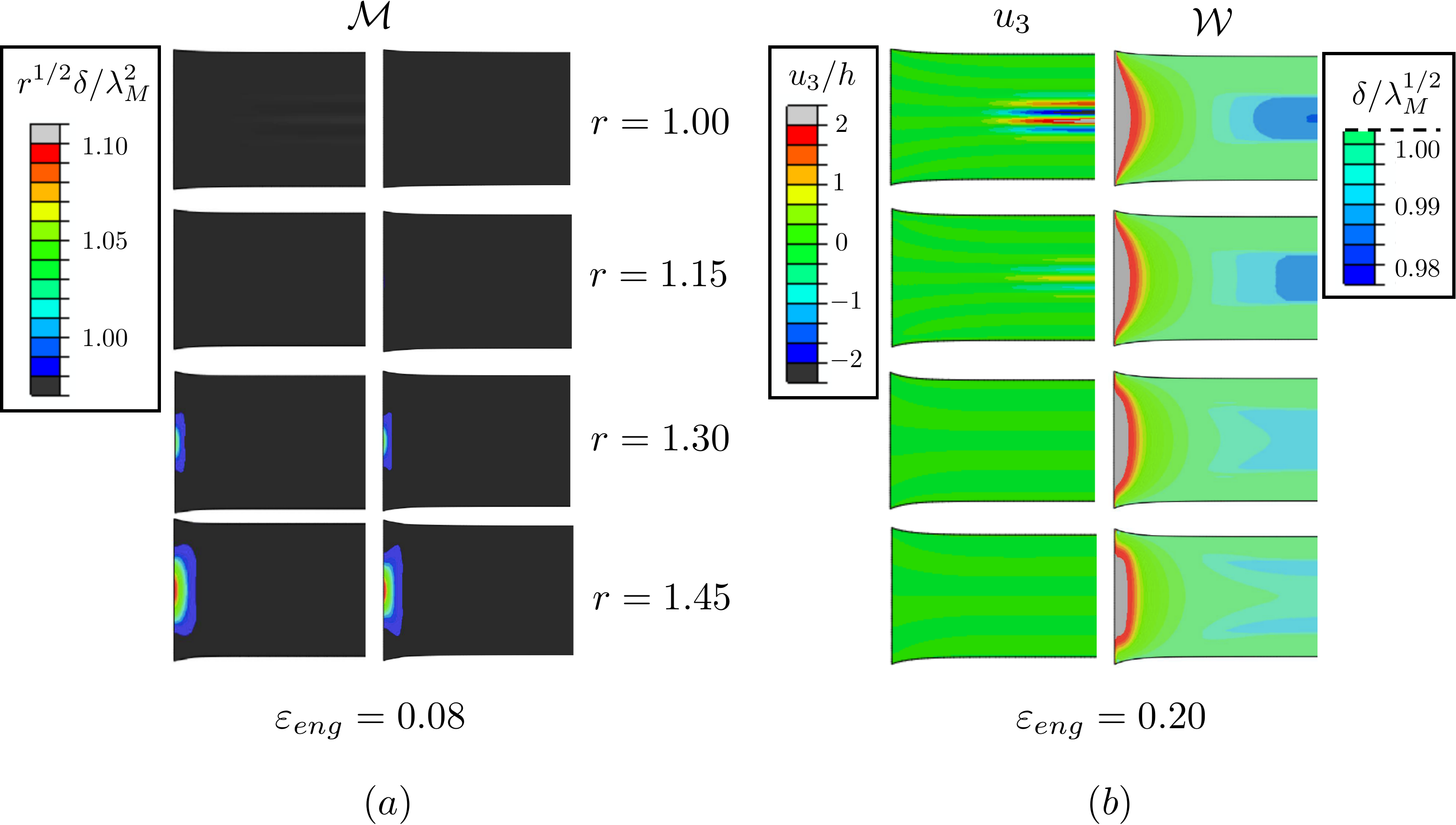}
\caption{Comparison of clamped stretched simulation under the effective membrane theory and Koiter theory for $r = 1,1.15,1.30$ and $1.45$ nematic sheets.  (a) Microstructure comparison at stretch $\varepsilon_{eng} = 0.08$; Koiter (left) and membrane (right).  (b) Wrinkling comparison at stretch $\varepsilon_{eng} =0.20$; Koiter (left) and membrane (right).  }
\label{fig:Comparison}
\end{figure}

We repeat the clamped-stretch simulations of purely elastic and nematic sheets using the membrane or tension field theory.  The nominal force vs. nominal strain is shown in Figure  \ref{fig:Soft}(a) as the points.  We see that the results agree well with those obtained using the Koiter theory.  We show the evolution of the microstructure at $\varepsilon_{eng} = 0.08$ for various values of $r$, and compare it with those obtained with the Koiter theory in Figure \ref{fig:Comparison}(a).  Again we find striking agreement.  Recall that the membrane theory does not describe the details of the wrinkles but relaxes over them, i.e., computes their consequence assuming that they were infinitely fine.  So, the results do not show any out of plane displacement.  However, recall that the curve $\delta = \lambda_M^{1/2}$ is the boundary between the presence and absence of wrinkling.  Therefore, any value of the parameter
\begin{equation}
\frac{\delta}{\lambda_M^{1/2}}
\end{equation}
below 1 indicates the presence of (infinitely) fine scale wrinkles in the simulation with this membrane theory.  This indicator is shown in Figure \ref{fig:Comparison}(b) for $\varepsilon_{eng} = 0.20$, and compared with the wrinkles computed earlier.  Again, we see very good agreement.  

Thus we conclude that the membrane or tension field theory provides a very good description of the overall behavior of the sheets including nominal stress-strain relation, the formation of microstructure and the formation of wrinkles.  Further, the agreement using an independent set of simulations provides confidence in our understanding regarding the suppression of wrinkles by nematic microstructure.

\section{Conclusions} \label{sec:conc}

We have studied the mechanical behavior of thin sheets of nematic elastomers subjected to stretch.   Nematic elastomers can form fine scale microstructure or stripe domains, whereas thin sheets can suffer wrinkling instabilities.  We have systematically developed a theory that is able to describe both instabilities.  Specifically, since  the scale of the microstructure is small compared to the scale of the wrinkles, we relax the microstructure but regularize the wrinkles by identifying the nature of the bending energy.  The result is a Koiter-type theory that includes a term to describe the in-plane stretch and a term to capture the out-of-plane deformation of initially flat sheets.

We have used this theory to show that the ability to form microstructure suppresses the wrinkles in nematic sheets subjected to stretch through clamped grips.  These results are consistent with the observations of Kundler and Finkelmann \cite{kf_mrc_95}.  We believe that this finding is technologically significant because wrinkling is a significant impediment in light-weight deployable space structures including solar sails, telescopes and antennas.  This remains a topic for future exploration.

Finally, our results show that the membrane or tension field theory of Cesana {\it {\it et al.}\ }\cite{cpk_arma_15} is able to capture the overall details including the macroscopic response, the ability of the material to form microstructure and the propensity of the sheets to wrinkle.

%%%%%%%%%%%%%%%%%%%%%%%%%%%%%%%%%%%%%%%%%%%%%%%%
%%%%%%%%%%%%%%%%%%%%%%%%%%%%%%%%%%%%%%%%%%%%%%%%
\section*{Acknowledgment}
P.P. is grateful for the support of the NASA Space Technology Research Program.

%%%%%%%%%%%%%%%%%%%%%%%%%%%%%%%%%%%%%%%%%%%%%%%%
%%%%%%%%%%%%%%%%%%%%%%%%%%%%%%%%%%%%%%%%%%%%%%%%

\appendix 

\section{The membrane term of the Koiter theory}\label{sec:KoiterAppend1}

\begin{proposition}\label{W3Dqc2D}
The equalities asserted in (\ref{eq:DeriveStretch}) and (\ref{eq:stretchEnergy}) do, in fact, hold.  Moreover, the infimum in (\ref{eq:DeriveStretch}) is attained for each full-rank $\tilde{F} \in \mathbb{R}^{3\times2}$ by setting
\begin{align}\label{eq:infimumAttained}
b^{\ast} := \arg \min_{b\in \mathbb{R}^3} W_{3D}^{qc}(\tilde{F}|b) = \frac{\tilde{F} \tilde{e_1} \times \tilde{F} \tilde{e_2}}{|\tilde{F} \tilde{e_1} \times \tilde{F} \tilde{e_2}|^2}.
\end{align}
\end{proposition}
\begin{proof}
Let $\tilde{F} \in \mathbb{R}^{3\times2}$ and $b \in \mathbb{R}^3$.  We may assume that $\tilde{F}$ is full-rank as the equality holds trivially otherwise (with $W_{3D}^{qc} = (W_{3D}^{qc})_{2D} = +\infty$).  By the singular value decomposition, we can set 
\begin{align}\label{eq:SVD1}
\tilde{F} = Q \tilde{D} \tilde{R}, \quad \text{ for } \quad Q \in SO(3), \quad \tilde{R} \in O(2),\quad  \tilde{D} = \text{diag}(\bar{\lambda}_M, \bar{\delta}/\bar{\lambda}_M) \in \mathbb{R}^{3\times2}.
\end{align}
Here $\bar{\lambda}_M = \lambda_M(\tilde{F}) > 0$ is the maximum singular value of $\tilde{F}$, and $\bar{\delta} =\delta(\tilde{F}) >0$ is the areal stretch of $\tilde{F}$ as defined below (\ref{eq:W2Dqc}).  Both are positive since $\tilde{F}$ is full-rank.  In addition, we let 
\begin{align}\label{eq:bijectB}
\bar{b} := (\det \tilde{R}) Q^T b.
\end{align}
By the frame invariance and isotropy of $W_{3D}^{qc}$, we observe that 
\begin{align}\label{eq:W3DqcSimple}
W_{3D}^{qc}(\tilde{F}|b) = W_{3D}^{qc}(Q(\tilde{D}|\bar{b})(\tilde{R}|(\det \tilde{R}) e_3)) = W_{3D}^{qc}(\tilde{D}|\bar{b})
\end{align}
since $Q$ and $(\tilde{R}|(\det \tilde{R}) e_3)$ are both in $SO(3)$.  Given this equality and recalling that $\tilde{F}$ is an arbitrary full-rank $\mathbb{R}^{3\times2}$ matrix and $b$ is arbitrary, we see that to prove the equality (\ref{eq:stretchEnergy}) it suffices to optimize $W_{3D}^{qc}$ amongst the deformation 
\begin{align}\label{eq:optimizeF}
F(\bar{b}_1,\bar{b}_2) := \left(\begin{array}{ccc} \bar{\lambda}_M & 0 & \bar{b}_1 \\ 0 & \bar{\delta}/\bar{\lambda}_M & \bar{b}_2 \\ 0 & 0 & \bar{\delta}^{-1} \end{array}\right).
\end{align}
Indeed, since $\bar{b}$ in (\ref{eq:bijectB}) is a bijective function of $b$, we have given (\ref{eq:W3DqcSimple})
\begin{align}\label{eq:InfChains}
(W_{3D}^{qc})_{2D}(\tilde{F}) := \inf_{b \in \mathbb{R}^3} W_{3D}^{qc}(\tilde{F}|b) = \inf_{\bar{b} \in \mathbb{R}^3} W_{3D}^{qc}(\tilde{D}|\bar{b}) = \inf_{(\bar{b}_1, \bar{b}_2) \in \mathbb{R}^2} W_{3D}^{qc}(F(\bar{b}_1,\bar{b}_2)).
\end{align}
For the last equality, as the infimum is necessarily incompressible, we enforce $\det( \tilde{D} |\bar{b}) = 1$ (which is only true if $\bar{b} \cdot e_3 = \delta^{-1}$ as in (\ref{eq:optimizeF})).  

Now, we claim that for any full-rank $\tilde{F}$,
\begin{align}\label{eq:qcClaim}
(W_{3D}^{qc})_{2D}(\tilde{F}) = W_{3D}^{qc}(F_0) \quad \text{ where } \quad F_{0} := F(0,0).
\end{align}
To show this, we will make use of the fact that $W_{3D}^{qc}(F) = \psi_{3D}(\lambda_M(F), \lambda_M(\cof F))$ and the fact that $\psi_{3D}$ is a monotonically increasing function in both of its arguments (this can be seen from the contour plot in Figure \ref{fig:3DEnergies}(b)).   In this direction, we first observe that 
\begin{align}\label{eq:qcIneq1}
\lambda^2_M(F(\bar{b}_1, \bar{b}_2)) \geq \max_{i \in \{1,2,3\}} |F(\bar{b}_1, \bar{b}_2) e_i|^2 \geq \max_{i \in \{1,2,3\}} |F_0 e_i|^2 = \lambda^2_M(F_0).
\end{align}
This is deduced from the parameterization in (\ref{eq:optimizeF}); as is the fact that 
\begin{align}\label{eq:cofFbarbs}
\cof F(\bar{b}_1, \bar{b}_2) =\left(\begin{array}{ccc} \bar{\lambda}_M^{-1} & 0 & 0 \\ 0 & \bar{\lambda}_M \bar{\delta}^{-1} & 0 \\ -\bar{b}_1 \bar{\delta}/\bar{\lambda}_M & -\bar{\lambda}_M \bar{b}_2 & \delta \end{array}\right).
\end{align}
With (\ref{eq:cofFbarbs}), we find that 
\begin{align}\label{eq:qcIneq2}
\lambda^2_M(\cof F(\bar{b}_1, \bar{b}_2)) \geq \max_{i \in \{1,2,3\}} |\cof F(\bar{b}_1, \bar{b}_2) e_i|^2 \geq \max_{i \in \{1,2,3\}} |\cof F_0 e_i|^2 = \lambda^2_M(\cof F_0)
\end{align}
as well.  Finally, combining (\ref{eq:qcIneq1}) and (\ref{eq:qcIneq2}) and using the fact that $\psi_{3D}$ is a monotonically increasing function in both its arguments, we find that 
\begin{equation}
\begin{aligned}
W^{qc}_{3D}(F(\bar{b}_1, \bar{b}_2)) &= \psi_{3D}(\lambda_M(F(\bar{b}_1,\bar{b}_2)), \lambda_M(\cof F(\bar{b}_1, \bar{b}_2)))\\
& \geq \psi_{3D}(\lambda_M(F_0), \lambda_M(\cof F_0)) = W_{3D}^{qc}(F_0).
\end{aligned}
\end{equation}
Hence, the infimum is attained by setting $(\bar{b}_1, \bar{b}_2) = 0$, i.e., 
\begin{align}\label{eq:qcClaimIdent}
\inf_{(\bar{b}_1,\bar{b}_2) \in \mathbb{R}^2} W^{qc}_{3D}(F(\bar{b}_1, \bar{b}_2)) = W_{3D}^{qc}(F_0),
\end{align}
and combining (\ref{eq:qcClaimIdent}) with (\ref{eq:InfChains}), we deduce the claim in (\ref{eq:qcClaim}).

The formulas (\ref{eq:WpsDef}) and (\ref{eq:stretchEnergy}) are obtained using (\ref{eq:qcClaim}) by explicit verification: briefly, if $\bar{\lambda}_M \geq \bar{\delta}^{-1}$ and $\bar{\delta} \geq \bar{\lambda}_M \bar{\delta}^{-1}$, then $(\lambda_M(F_0), \lambda_M(\cof F_0)) = (\bar{\lambda}_M, \bar{\delta})$ and 
\begin{align}\label{eq:qcExplicit1}
W_{3D}^{qc}(F_0) = \psi_{3D}(\bar{\lambda}_M, \bar{\delta}) = \psi_{2D}(\bar{\lambda}_M, \bar{\delta}) = W_{2D}^{qc}(\tilde{F}) , \quad (\lambda_M(\tilde{F}), \delta(\tilde{F})) \in \mathcal{M} \cup \mathcal{S} \cup (\mathcal{L}_M \cap L);
\end{align}
if $\bar{\lambda}_M \geq \bar{\delta}^{-1}$ and $\bar{\lambda}_M \bar{\delta}^{-1} > \bar{\delta}$, then $(\lambda_M(F_0), \lambda_M(\cof F_0)) = (\bar{\lambda}_M, \bar{\lambda}_M \bar{\delta}^{-1})$ and 
\begin{equation}
\begin{aligned}\label{eq:qcExplicit2}
W_{3D}^{qc}(F_0) = \psi_{3D}(\bar{\lambda}_M, \bar{\lambda}_M \bar{\delta}^{-1}) = \begin{cases}
\psi_{2D}(\bar{\lambda}_M, \bar{\delta}) = W_{2D}^{qc}(\tilde{F}) & \text{ if } (\lambda_M(\tilde{F}), \delta(\tilde{F}) )\in \mathcal{L}_m \setminus L \\
\varphi_{2D}(\bar{\lambda}_M, \bar{\delta}) = W_{2D}(\tilde{F}) & \text{ if }  (\lambda_M(\tilde{F}), \delta(\tilde{F})) \in \mathcal{C} \cap \tilde{\mathcal{C}}
\end{cases}
\end{aligned}
\end{equation}
where $\tilde{\mathcal{C}} := \{ (s,t) \in \mathbb{R}^+ \times \mathbb{R}^+ \colon t \geq s^{-1}\}$; if $\bar{\lambda}_M < \bar{\delta}^{-1}$ and $\bar{\lambda}_M \bar{\delta}^{-1} > \bar{\delta}$, then $(\lambda_M(F_0), \lambda_M(\cof F_0)) = (\bar{\delta}^{-1}, \bar{\lambda}_M \bar{\delta}^{-1})$ and 
\begin{align}\label{eq:qcExplicit3}
W_{3D}^{qc}(F_0) = \psi_{3D}(\bar{\delta}^{-1}, \bar{\lambda}_M \bar{\delta}^{-1})  = \varphi_{2D}(\bar{\lambda}_M, \bar{\delta}) = W_{2D}(\tilde{F}), \quad (\lambda_M(\tilde{F}), \delta(\tilde{F})) \in \mathcal{C} \setminus \tilde{\mathcal{C}}.
\end{align}
This exhausts all possible cases, as any other case is actually incompatible with the fact that $\bar{\delta} = \bar{\lambda}_M \bar{\lambda}_m$ where $\bar{\lambda}_m = \lambda_m(\tilde{F})$ is the minimum singular value of $\tilde{F}$.  Combining (\ref{eq:qcExplicit1}), (\ref{eq:qcExplicit2}) and (\ref{eq:qcExplicit3}) with the fact that (\ref{eq:qcClaim}) holds, we arrive at the representation in (\ref{eq:stretchEnergy}).  

To see that the infimum is attained with $b^{\ast}$ in (\ref{eq:infimumAttained}), we notice that actually 
\begin{align}
b^{\ast} = \bar{\delta}^{-2} Q (\tilde{D} \tilde{R} \tilde{e}_1 \times \tilde{D} \tilde{R} \tilde{e}_2) = \bar{\delta}^{-1}(\det \tilde{R}) Qe_3.
\end{align}
for $Q, \tilde{D}$ and $\tilde{R}$ in (\ref{eq:SVD1}).  Hence, again using the frame invarience and isotropy of $W_{3D}^{qc}$, 
\begin{align}
W_{3D}^{qc}(\tilde{F}|b^{\ast}) = W_{3D}^{qc}(Q(\tilde{D}| \bar{\delta}^{-1} e_3)(\tilde{R}| (\det \tilde{R}) e_3)) = W_{3D}^{qc}(F_0).  
\end{align}
So given (\ref{eq:qcClaim}), the infimum is attained with $b^{\ast}$ as desired.  This completes the proof.  
\end{proof}

%%%%%%%%%%%%%%%%%%%%%%%%%%%%%%%%%%%%%%%%%%%%%%%%
%%%%%%%%%%%%%%%%%%%%%%%%%%%%%%%%%%%%%%%%%%%%%%%%
\section{Extension of the midplane deformation to the entire sheet}\label{sec:KoiterAppend}

%%%%%%%%%%%%%%%%%%%%%%%%%%%%%%%%%%%%%%%%%%%%%%%%
\begin{proposition}\label{IncompProp}
Fix $\tau > 0$.  There exists an $\bar{h} = \bar{h}(\tau) > 0$ with $\bar{h}k < 1$ such that for any $y_k$ in (\ref{eq:ykDef}), $b_k$ in (\ref{eq:bkDef}) and $h \in (0,\bar{h})$, there exists a unique $\xi_k^h \in C^1(\bar{\Omega}_h,\mathbb{R}^3)$ for which $y_k^h$ defined in (\ref{eq:GlobalIncomp1}) satisfies (\ref{eq:GlobalIncomp2}) and $\nabla y_k^h$ satisfies (\ref{eq:ImportantForm}).
\end{proposition}

\begin{proof}
We first consider the naive deformation 
\begin{align}
v_k^h(x) := y_k(\tilde{x}) + x_3 b_k(\tilde{x}), \quad \tilde{x} \in \omega.
\end{align}
This deformation is not incompressible.  However, it is nearly so, as the gradient
\begin{align}\label{eq:vkhGrad}
\nabla v_k^h = (\tilde{\nabla} y_k |b_k) + x_3 (\tilde{\nabla} b_k|0) \quad \text{ on } \Omega_h
\end{align}
has a determinant which is $O(kx_3)$ close to unity.  Indeed, $\det(\tilde{\nabla} y_k |b_k) = 1$ by construction, and so we define 
\begin{align}\label{eq:GkDef}
G_k := (\tilde{\nabla} y_k|b_k)^{-1} (\tilde{\nabla} b_k|0) \quad \text{ on } \omega,
\end{align}
 and observe that given (\ref{eq:vkhGrad}) and (\ref{eq:GkDef}) and since $\det((\tilde{\nabla} y_k |b_k)^{-1}) = 1$ on $\omega$, 
\begin{equation}
\begin{aligned}\label{eq:detIdentity}
\det(\nabla v_k^h) &= \det( (\tilde{\nabla} y_k|b_k)^{-1} \nabla v_k^h) = \det( I + x_3G_k) = 1 + x_3 \Tr(G_k)  \quad \text{ on } \Omega_h . 
\end{aligned}
\end{equation}
Here, we used that both $\Tr(\cof G_k)$ and $\det G_k $ are zero on $\omega$.  With this relationship, we find that  
\begin{equation}
\label{eq:detEsts}
\begin{aligned}
&|\det(\nabla v_k^h) -1| \leq C(\tau) k|x_3|, \quad \text{ on } \Omega_h, \quad kh < 1, \\
&|\partial_2 \det(\nabla v_k^h)| \leq  C(\tau) k^2 |x_3|, \quad \text{ on } \Omega_h, \quad kh < 1, \\
& \partial_1 \det(\nabla v_k^h) = 0 \quad \text{ on } \Omega_h,
\end{aligned}
\end{equation}
using the fact that $G_k$ is independent of $x_1$ and $\gamma_k \in \mathcal{A}_k^\tau$.  An interesting observation is that as implied, the constant $C(\tau) >0$ can be chosen to depend only on $\tau$.  In particular, $G_k$ does not depend appreciably on $\bar{\lambda}_M > r^{1/3}$, and so the constant need not depend on this stretch.

Now, as (\ref{eq:detIdentity}) shows, the determinant of a generic $\nabla v_k^h$ is not unity.  Thus, we introduce a $\xi_k^h \in C^1(\bar{\Omega}_h, \mathbb{R})$  (for some $h>0$ to be chose later) satisfying $\xi_k^h(\tilde{x},0)= 0$ and define $y_k^h$ as in (\ref{eq:GlobalIncomp1}).  Using properties of the determinant, we find that 
\begin{align}\label{eq:ode}
\det (\nabla y_k^h )= 1 \quad \text{ on } \Omega_h \quad \Leftrightarrow \quad \partial_3 \xi_k^h(x) = \frac{1}{\det \left(\nabla v_k^h(\tilde{x}, \xi_k^h(x)) \right)}, \quad x \in \Omega_h.  
\end{align}
Therefore, the incompressibility of $\nabla y_k^h$ is equivalent to solving an ordinary differential equation in $\xi_k^h(\tilde{x}, \cdot)$.  It turns out that there exists an $\bar{h} = \bar{h}(\tau)> 0$ such that for any $h \in (0,\bar{h})$ and $k$ such that $kh < 1$, there exists a unique $\xi_k^h \in C^{1}(\bar{\Omega}_h,\mathbb{R})$ satisfying this differential equation subject to the initial condition $\xi_k^h(\tilde{x}, 0) = 0$.  Moreover, from (\ref{eq:detEsts}), $\xi_k^h$ has the properties that 
\begin{equation}
\begin{aligned}\label{eq:xikhProps}
&|\xi_k^h - x_3| \leq C(\tau)k|x_3|^2, \quad |\partial_3 \xi_k^h - 1| \leq C(\tau) k|x_3|, \quad \text{ on } \Omega_h\\
&|\partial_2 \xi_k^h| \leq C(\tau)k^2 |x_3|^2, \quad \partial_1 \xi_k^h = 0, \quad  \text{ on }  \Omega_h.
\end{aligned}
\end{equation}
This result is the consequence of a contraction map principle and some further analysis (presented, for instance, in \cite{plb_arxiv_16} Section 3). 

In the remainder, we assume $h \in (0,\bar{h})$, $kh < 1$, and $y_k^h$ as in (\ref{eq:GlobalIncomp1}) for $\xi_k^h \in C^{1}(\bar{\Omega}_h, \mathbb{R})$ satisfying $\xi_k^h(\tilde{x},0) = 0$, the ordinary differential equation (\ref{eq:ode}) and the estimates (\ref{eq:xikhProps}).  By explicit calculation, we have 
\begin{equation}
\label{eq:gradykh}
\begin{aligned}
\nabla y_k^h &= (\tilde{\nabla} y_k|b_k) + x_3(\tilde{\nabla} b_k|0)  + (\partial_3 \xi_k^h - 1) b_k \otimes e_3  \\
&\quad + (\xi_k^h - x_3) (\tilde{\nabla} b_k|0) + b_k \otimes \tilde{\nabla} \xi_k^h \quad \text{ on } \Omega_h. 
\end{aligned}
\end{equation}

We can extract a more illuminating form by examining closely identities and estimates for $\xi_k^h$.  
Indeed, combining the ordinary differential equation (\ref{eq:ode}) with the parameterization of $\det(\nabla v_k^h)$ in (\ref{eq:detIdentity}), we find that 
\begin{equation}
\label{eq:xikh1}
\begin{aligned}
\partial_3 \xi_k^h - 1 &= - \partial_3 \xi_k^h \xi_k^h \Tr(G_k)  \quad \text{ on } \Omega_h.
\end{aligned}
\end{equation}
In addition, making use of the boundary condition $\xi_k^h(\tilde{x},0) = 0$, the fundamental theorem of calculus, and our newfound parameterization (\ref{eq:xikh1}),
\begin{equation}
\label{eq:xikh2}
\begin{aligned}
\xi_k^h - x_3 &= -\frac{1}{2} (\xi_k^h)^2 \Tr(G_k)  \quad \text{ on } \Omega_h.
\end{aligned}
\end{equation}
Hence, with the estimates (\ref{eq:xikhProps}) and these parameterizations, we establish that 
\begin{equation}
\label{eq:xikh3}
\begin{aligned}
\partial_3 \xi_k^h - 1 &= -x_3 \Tr(G_k) + \frac{3}{2} x_3^2 \Tr(G_k)^2  + O(k^3 x_3^3)  \quad \text{ on } \Omega_h, \\
\xi_k^h - x_3 &=-\frac{1}{2} x_3^2 \Tr(G_k) + O(k^2 x_3^3) \quad \text{ on } \Omega_h.
\end{aligned}
\end{equation}
Finally, the results on the planar derivatives of $\xi_k^h$ in (\ref{eq:xikhProps}) imply that 
\begin{equation}
\label{eq:xikh4}
\begin{aligned}
b_k \otimes \tilde{\nabla} \xi_k^h & = O(k^2 x_3^2) b_k \otimes e_2 \quad \text{ on } \Omega_h.
\end{aligned}
\end{equation}
Combining (\ref{eq:gradykh}), (\ref{eq:xikh2}), (\ref{eq:xikh3}) and (\ref{eq:xikh4}), we arrive at the parameterization for $\nabla y_k^h$ in (\ref{eq:ImportantForm}) as desired.  
\end{proof}

%%%%%%%%%%%%%%%%%%%%%%%%%%%%%%%%%%%%%%%%%%%%%%%%
\begin{proposition}\label{EigenProp}
Let $y_k^h$ be as in Proposition \ref{IncompProp} under the assumption therein.  (\ref{eq:IdentLambdaM}) holds and if $kh$ is sufficiently small, then (\ref{eq:IdentLambdaM2}) also holds.  
\end{proposition}
\begin{proof}
The functions $\{ e_1,\gamma_k', \nu_{y_k}\}$ form an orthonormal basis of $\mathbb{R}^3$ and $b_k = \bar{\lambda}_M^{-1/2} \nu_{y_k}$ and is independent of $x_1$.  Here $\nu_{y_k}$ is the surface normal of $y_k$ as defined in (\ref{eq:secFundForm}).  Further, $y_k^h$ is as in Proposition \ref{IncompProp}, so in particular the representations (\ref{eq:gradykh}) and (\ref{eq:ImportantForm}) hold.  Combining all this, we conclude that 
\begin{equation}
\label{eq:eigenvectorCalc}
\begin{aligned}
(\nabla y_k^h)^T e_1 &= (\tilde{\nabla} y_k |b_k)^T e_1 =  \bar{\lambda}_M  e_1  \quad \text{ on } \Omega_h \\
(\nabla y_k^h)^T \gamma_k' &= (\tilde{\nabla} y_k|b_k)^T \gamma_k' + O(x_3) (\tilde{\nabla} b_k|0)^T \gamma_k' = \left( \bar{\lambda}_M^{-1/2}  + O(kx_3)\right)e_2  \quad \text{ on } \Omega_h\\
(\nabla y_k^h)^T \nu_{y_k} &= (\tilde{\nabla} y_k |b_k)^T \nu_{y_k} + O(kx_3) (e_3 \otimes b_k) \nu_{y_k} + O(k^2x_3^2) (e_2 \otimes b_k) \nu_{y_k} \\
&=\left(\bar{\lambda}_M^{-1/2} + O(kx_3)\right)e_3 + O(k^2 x_3^2) e_2 \quad \text{ on } \Omega_h
\end{aligned}
\end{equation}
where the first result uses that $e_1$ is orthogonal to $b_k$ and $\tilde{\nabla} b_k$, the second result uses that $\gamma_k'$ is orthogonal to $b_k$ and the last result uses that $\nu_{y_k}$ is orthogonal the $\tilde{\nabla} b_k$.  Since $\bar{\lambda}_M > \bar{\lambda}_M^{-1/2}$, the first term dominates the other two terms for $kh$ sufficiently small, and we readily conclude that
\begin{align}\label{eq:EigEqual}
\lambda_M(\nabla y_k^h) = \lambda_M(\tilde{\nabla} y_k|b_k) = \bar{\lambda}_M > r^{1/3}
\end{align}
as desired. 

Now to show the $\lambda_M( \cof (\tilde{\nabla} y_k|b_k)) = \bar{\lambda}_M^{-1/2}$, we note that $\lambda_M(\cof F) = \lambda_M(F) \lambda_2(F)$ for any $F \in \mathbb{R}^{3\times3}$ where $\lambda_2$ denotes the middle singular value.  Thus, the result follows from the fact that on $\omega$, $(\tilde{\nabla} y_k|b_k)^T \gamma_k' = \bar{\lambda}_M^{-1/2} e_2 $, $(\tilde{\nabla} y_k|b_k)^T \nu_{y_k} = \bar{\lambda}_M^{-1/2} e_3$ and $(\tilde{\nabla} y_k|b_k)^T e_1 = \bar{\lambda}_M e_1$.  This proves (\ref{eq:IdentLambdaM}).

Now given (\ref{eq:EigEqual}) we know that $\nabla y_k^h$ cannot lie in $L$ on the energy landscape for $W_{3D}^{qc}$.  So to complete the proof, we simply have to conclude that $\nabla y_k^h$ is not in $M$ on this landscape.  This is assured if 
\begin{align}\label{eq:notM}
\bar{\lambda}_M > r^{1/2} \lambda_2(\nabla y_k^h),
\end{align}
again given (\ref{eq:EigEqual}) and the definition of the set $M$ in (\ref{eq:setM}).  It is easy to see that from (\ref{eq:eigenvectorCalc}) that
\begin{align}\label{eq:lipschitz2}
\lambda_2(\nabla y_k^h)  = \bar{\lambda}_M^{-1/2} + O(kx_3).
\end{align}
In addition, since $\bar{\lambda}_M > r^{1/3}$, we know that $\bar{\lambda}_M > r^{1/2} \bar{\lambda}_M^{-1/2}$.  Thus, (\ref{eq:notM}) follows from (\ref{eq:lipschitz2}) for $kh$ sufficiently small, and therefore $\nabla y_k^h$ lies in $S$ on the energy landscape.  This completes the proof. 
\end{proof}

%%%%%%%%%%%%%%%%%%%%%%%%%%%%%%%%%%%%%%%%%%%%%%%%
\begin{proposition}\label{PropIdentities}
For $y_k$ as in (\ref{eq:ykDef}), $b_k$ as in (\ref{eq:bkDef}) and $G_k$ as in (\ref{eq:ImportantForm}) we have the following pointwise identities everywhere on $\omega$:
\begin{align}
&(\tilde{\nabla} y_k |b_k)^T b_k = |b_k|^2 e_3 = \bar{\lambda}_M^{-1} e_3, \label{eq:PropIdent1} \\
&(\tilde{\nabla} b_k|0)^T b_k = 0,\label{eq:PropIdent2} \\
&|\tilde{\nabla} b_k|^2 = |\II_{y_k}|^2 \label{eq:PropIdent6} \\
&\Tr((\tilde{\nabla} y_k |b_k)^T(\tilde{\nabla} b_k|0)) = \bar{\lambda}_M^{-1/2}\Tr(\II_{y_k}), \label{eq:PropIdent7} \\
&\Tr(G_k) = \bar{\lambda}_M^{1/2}\Tr( \II_{y_k}), \label{eq:PropIdent8}\\
&\Tr(\II_{y_k})^2 = |\II_{y_k}|^2. \label{eq:PropIdent9}
\end{align}
\end{proposition}
\begin{proof}
First, we observe that $b_k = \bar{\lambda}_M^{-1/2} \nu_{y_k}$ where $\nu_{y_k}$ is the surface normal of $y_k$.  Thus, $(\tilde{\nabla}y_k)^T b_k = 0$ by definition, and (\ref{eq:PropIdent1}) follows.  In addition, since $\nu_{y_k}$ is a unit vector, $\partial_{\alpha} \nu_{y_k} \cdot \nu_{y_k} = (1/2) \partial_{\alpha} |\nu_{y_k}|^2 = 0$.  So (\ref{eq:PropIdent2}) follows.  

Now, for the third equality, we first observe that for tension wrinkling
\begin{align}\label{eq:Fund1}
\II_{y_k} = (\tilde{\nabla} y_k)^T (\tilde{\nabla} \nu_{y_k}) = \left(\begin{array}{cc} \partial_1 y_k \cdot \partial_1 \nu_{y_k} & \partial_1 y_k \cdot \partial_2 \nu_{y_k} \\  \partial_2 y_k \cdot \partial_1 \nu_{y_k} & \partial_2 y_k \cdot \partial_2 \nu_{y_k} \end{array}\right) = \left(\begin{array}{cc} 0 & 0 \\ 0 & \partial_2 y_k \cdot \partial_{2} \nu_{y_k} \end{array} \right)
\end{align}
using the fact that $\nu_{y_k}$ is independent of $x_1$ and that $\partial_1 y_k \cdot \partial_2 \nu_{y_k} = \bar{\lambda}_M e_1 \cdot \partial_2 \nu_{y_k} = 0$ since $\nu_{y_k} \cdot e_1 = 0$.  Thus, the second fundamental form is greatly simplified in the context of tension wrinkling.  In addition, we observe that 
\begin{equation}
\label{eq:Fund2}
\begin{aligned}
&(\tilde{\nabla} y_k |b_k) = R_k \Lambda_M,  \quad \Lambda_M = \text{diag}( \bar{\lambda}_M, \bar{\lambda}_M^{-1/2}, \bar{\lambda}_M^{-1/2} )\\
&R_k := \left(\begin{array}{ccc} 1 & 0 & 0 \\ 0 & (\gamma_k' \cdot e_2) & -(\gamma_k' \cdot e_3) \\ 0 & (\gamma_k' \cdot e_3) & (\gamma_k' \cdot e_2) \end{array}\right) \in SO(3).
\end{aligned}
\end{equation}
Consequently,
\begin{align}\label{eq:Fund3}
|\tilde{\nabla} b_k|^2 = |R_k^T (\tilde{\nabla} b_k)|^2 = |\Lambda_M^{-1} (\tilde{\nabla} y_k|b_k)^T \tilde{\nabla} b_k|^2 = (\partial_2 y_k \cdot \partial_2 \nu_{y_k})^2.  
\end{align}
Combining (\ref{eq:Fund1}) and (\ref{eq:Fund3}), we obtain (\ref{eq:PropIdent6}) as desired.  The calculation for (\ref{eq:PropIdent7}) is similar.  For (\ref{eq:PropIdent8}), we can deduce from (\ref{eq:Fund2}) that 
\begin{align}
(\tilde{\nabla} y_k |b_k)^{-1} = \Lambda_M^{-1} R_k^T = \Lambda_M^{-2} (\tilde{\nabla} y_k |b_k)^T.  
\end{align}
Consequently,
\begin{align}
\Tr(G_k ) = \Tr(\Lambda_M^{-2} (\tilde{\nabla} y_k |b_k)^T(\tilde{\nabla} b_k|0))  = \bar{\lambda}_M^{1/2} (\partial_2 y_k \cdot \partial_2 \nu_{y_k}),
\end{align}
and (\ref{eq:PropIdent8}) follows from (\ref{eq:Fund1}).  Finally, (\ref{eq:PropIdent9}) also follows from (\ref{eq:Fund1}).  This completes the proofs.  
\end{proof}

\bibliographystyle{abbrv}
\bibliography{pbWrinkMicro}

\end{document}